\documentclass[twocolumn,showpacs,preprintnumbers,amsmath,amssymb]{revtex4-1}

\usepackage{graphicx}
\usepackage{dcolumn}
\usepackage{bm}
\usepackage{soul, xcolor}
\usepackage{color, wasysym}
\usepackage{textcomp}
\usepackage{amsmath, amssymb, amsthm, latexsym, epsfig}
\usepackage{mathrsfs}
\usepackage[T2A,T1]{fontenc}
\usepackage[utf8]{inputenc}
\usepackage[russian,english]{babel}

\newtheorem{lemma}{Lemma}
\newtheorem{proposition}{Proposition}
\newtheorem{remark}{Remark}

\newtheorem{example}{Example}

\usepackage{appendix}

\def\ulamek#1#2{\mbox{\normalfont$\frac{#1}{#2}$}}

\DeclareMathOperator{\okr}{{\stackrel{{\scriptscriptstyle{\mathsf{def}}}}{=}}}
\DeclareMathOperator{\D}{d\!}
\DeclareMathOperator{\E}{e} 
\DeclareMathOperator{\I}{i}
   \DeclareMathOperator{\RE}{\mathfrak{Re}}
   \DeclareMathOperator{\IM}{\mathfrak{Im}}\def\ulamek#1#2{\mbox{\normalfont$\frac{#1}{#2}$}}

\numberwithin{equation}{section}

\begin{document}

\title{The Havriliak-Negami and Jurlewicz-Weron-Stanislavsky relaxation models revisited: memory functions based study}

\author{K. G\'{o}rska} 
\author{A. Horzela}
\email{katarzyna.gorska@ifj.edu.pl}
\email{andrzej.horzela@ifj.edu.pl}
\affiliation
{H. Niewodnicza\'{n}ski Institute of Nuclear Physics, Polish Academy of Sciences, 
ul.Eljasza-Radzikowskiego 152, PL 31342 Krak\'{o}w, Poland \vspace{2mm}}

\author{K. A. Penson}
\email{karol.penson@sorbonne-universite.fr}
\affiliation
{$^b$ Sorbonne Universit\'{e}, Campus Pierre et Marie Curie (Paris VI), CNRS UMR 7600,
Laboratoire de Physique Th\'{e}orique de la Mati\`{e}re Condens\'{e}e (LPTMC),
Tour 13, 5i\`{e}me \'{e}t., B.C. 121, 4 pl. Jussieu, F 75252 Paris Cedex 05, France\vspace{2mm}}

\begin{abstract}
We provide a review of theoretical results concerning the Havriliak-Negami (HN) and the Jurlewicz-Weron-Stanislavsky (JWS) dielectric relaxation models. We derive explicit forms of functions characterizing relaxation phenomena in the time domain - the relaxation, response and probability distribution functions. We also explain how to construct and solve relevant evolution equations within these models. These equations are usually solved by using the Schwinger parametrization and the integral transforms. Instead, in this work we replace it by the powerful Efros theorem.  That allows one to relate physically admissible solutions to the memory-dependent evolution equations with phenomenologically known spectral functions and, from the other side, with the subordination mechanism emerging from a stochastic analysis of processes  underpinning considered relaxation phenomena. Our approach is based on a systematic analysis of the memory-dependent evolution equations. It exploits methods of integral transforms, operational calculus and special functions theory with the completely monotone and  Bernstein functions. Merging analytic and stochastic methods enables us to give a complete classification of the standard functions used to describe the large class of the relaxation phenomena and to explain their properties.          
\end{abstract}


\maketitle

\section{Introduction}\label{SEC1}

Researchers working in the fields of natural and engineering sciences are used to understand relaxation phenomena as processes which describe effects related to the delay between the application of an external stress to a system and its response. Phenomena which exhibit such behavior are observed practically everywhere in our surroundings: decay of induced dielectric polarization or magnetization usually does not follow switching off the external electromagnetic field instantaneously and various luminescence phenomena persist in the absence of illumination. Similarly,  deformed viscoelastic materials return to their original shapes with shorter or longer delay, kinetics of chemical reactions in many cases is retarded with respect to changes of external conditions, and so on. The simplest mathematical model used to describe relaxation is the exponential, or the Debye \cite{F1}, law which provides us with the time decay rule proportional to the exponential function $\exp{(-\lambda t})$, $\lambda > 0$. The Debye law works well for different physical systems but suffers more and more difficulties or even fails with the growing complexity of the system under investigation. The first more systematic observations of non-exponential relaxations and/or decays  were performed more than one hundred seventy years ago, in the middle of the XIX-th century, by R. Kohlrausch \cite{Kohlrausch1854} who carried out a series of relaxation experiments. At first he was focused on the mechanical creep but soon after he became interested in electrical phenomena and successfully measured the residual discharge current in the Leyden jar. Studying his own experimental results Kohlrausch found that the data were much better described by the stretched exponent $\exp{(-\lambda t^{\alpha})}$, $\alpha\in(0,1)$, which means that the relaxation is slower than predicted by the Debye law.  During the next few decades, deviations from the Debye law were observed also in relaxation phenomena going far beyond mechanics of viscoelastic media and effects occurring in simple resistor-capacitor circuits. We mention the luminescence phenomena  \cite{BarberanSantos05,BarberanSantos08} or the response of dielectric material to the step input of a direct current voltage as well. Mathematical models tested to explain experimental results were, as a rule, based on the inverse power-like decay model. Examples include the Curie-von Schweidler law of current decay or Nutting's equation for the stress-strain relation, both usually treated as purely phenomenological input proposed without any further justification. This phenomenology dominated period of relaxation studies ended in the context of physics of dielectrics around the 60s and 70s of the previous century. With the progress of experimental techniques, new materials were more and more common applied and investigated. Seminal results were obtained by A. K. Jonscher via analysis of a vast majority of available dielectric relaxation data. He found regularities hidden behind them \cite{Jonscher77, Jonscher83, Jonscher96}. Jonscher's discoveries, unified in the so-called {\em Universal Relaxation Law} (URL), boil down to the statement that the asymptotics of relaxation functions, both in the frequency and the time domains, are governed by fractional power-like functions. Also at that time, more than a hundred years after Kohlrausch's experiments, the stretched exponential function got back in the game and found wider interest among physicists thanks to the work of G. Williams and D. C. Watts who applied the stretched exponential to fit the dielectric relaxation data in the time domain \cite{WW70}. It should be mentioned that despite its popularity among experimentalists, the use of the stretched exponent (now called also the Kohlrausch-Williams-Watts (KWW) function) has remained formal as it has never found a deeper theoretical explanation. Moreover, its use leads to an embarrassing situation: the KWW function, if transformed from the time to the frequency domain, is expressed by rather involved special functions, see Sect. \ref{sec3}.  The latter are defined either by special contour integrals in the complex plane or through slowly convergent power series. This makes it difficult to read out correctly the dependence of dielectric permittivities on the frequency of applied harmonic electric field, even if these relations, nick-named spectral functions, are not only well-known phenomenologically from the broadband dielectric spectroscopy experiments, but do appear to be simple rational functions \cite{CJFBottcher78}. For that reason, the KWW model is often replaced by the Cole-Cole (CC) \cite{CJFBottcher78, RHilfer02, RGarrappa16}, Cole-Davidson (CD) \cite{CJFBottcher78, RHilfer02, RGarrappa16}, Havriliak-Negami (HN) \cite{CJFBottcher78, RHilfer02, RGarrappa16}, and Jurlewicz-Weron-Stanislavsky (JWS) \cite{AStanislavsky10, KWeron10, RGarrappa16} patterns. Their spectral functions are much simpler than the KWW model but the opposite situation appears for the time dependence functions.

Contemporary ongoing experiments aimed at collecting results essential for a better understanding of dielectric relaxation phenomena may be classified in a twofold way. The data which are measured in a suitable experimental setup concern either the time dependence of physically meaningful properties characterizing the system (e.g., the number of previously excited dipoles which survive some amount of  time after sudden switching off the polarizing field or retreating mechanical deformation of the sample), or the frequency-dependent response of the system perturbed by a step-like or alternating (usually harmonic) external electric field which plays the role of time-dependent external stress. In the first case, the results are traditionally encoded in the time-dependent relaxation $n(t)$ (or the response $\phi(t)=-\D n(t)/\D t$) function, while in the second case, customarily used objects are complex-valued spectral functions $\widehat{\phi}(\I\!\omega)$. In dielectric physics, the spectral functions $\widehat{\phi}(\I\!\omega)$ mimic normalized complex dielectric permittivity $[\widehat{\varepsilon^{\star}}(\omega)-\varepsilon_{\infty}]/[\varepsilon_{0}-\varepsilon_{\infty}]$, whose real and imaginary parts are responsible for diffractive and absorptive properties of the medium, and the material constants $\varepsilon_{0}$ and $\varepsilon_{\infty}$ denote the low (static) and high-frequency values of the dielectric permittivity. The time and frequency approaches are mutually connected by the Laplace transform \cite{F2} $f(t)\div \widehat{f}(z) = \int_{0}^{\infty}\exp{(-zt)}f(t)\D t = {\mathscr L}[f(t),z]$ which boils down to the following relations between the spectral, response and relaxation functions:
\begin{equation}\label{23/05_1}
\widehat{\phi}(\I\!\omega) = \mathscr{L}[\phi(t); \I\!\omega] = 1-\I\!\omega\,\widehat{n}(\I\!\omega),
\end{equation}                  
where we used the fact that
\begin{equation}\label{16/06-20}
\phi(t) = - \frac{\D n(t)}{\D t} \quad \text{or, equivalently,} \quad n(t) = 1 - \int_{0}^{t} \phi(u) \D u.
\end{equation}
Equation \eqref{23/05_1} is crucial for our further consideration. It gives an operational rule which constitutes a bridge between the time and frequency descriptions of the relaxation phenomena and opens possibilities to compare them. This applies also if the data come from independent measurements performed in the time and frequency domains on the same, or twin-like prepared, sample. 

We shall now give more details on the phenomenological models specified by CC, CD, HN, and JWS earlier. In the frequency domain the phenomenological pattern relevant to the Debye model is equal to
\begin{equation}\label{9/02/23-1}
\widehat{\phi}_{D}(\I\!\omega) = (1 + \I\!\omega\tau)^{-1}
\end{equation}
and non-Debye relaxations are modelized through:
\begin{align}\label{23/05_2} 
\begin{split}
\widehat{\phi}_{CC}(\I\!\omega)& =[1+(\I\!\omega\tau)^{\alpha}]^{-1}, \\ \widehat{\phi}_{C\!D}(\I\!\omega) &=[1+(\I\!\omega\tau)]^{-\beta},\\ \widehat{\phi}_{H\!N}(\I\!\omega)& =[1+(\I\!\omega\tau)^{\alpha}]^{-\beta}, \\ \widehat{\phi}_{J\!W\!S}(\I\!\omega) &=1-[1+(\I\!\omega\tau)^{-\alpha}]^{-\beta}.
\end{split}
\end{align}
Real numbers $\alpha, \beta\in(0, 1]$ are called the width and the symmetry parameters, respectively, and each time they are fitted to the experimental data together with the material-dependent characteristic time scale $\tau$. It is easy to see that Eqs. \eqref{23/05_2} extend the Debye rule and that among them the HN and JWS patterns are the most general involving the CC and CD as their special cases \cite{CJFBottcher78, RGarrappa16}. The Debye and CC models emerge either from $\widehat{\phi}_{H\!N}(\I\!\omega)$ and $\widehat{\phi}_{J\!W\!S}(\I\!\omega)$ for $\alpha = \beta = 1$ or for $\beta = 1$ and $\alpha \in(0, 1)$, respectively. The CD model comes out from the HN pattern for $\alpha = 1$ and $\beta\in(0, 1)$. On the other hand, measurements done in the time domain provide us with the data for the relaxation function which are usually fitted by the standard KWW function 
\begin{equation}\label{23/05_3}
n_{KWW}(t) = \exp{[-(t/\tau^{\prime})^{\alpha^{\prime}}]}
\end{equation}
where parameters $\tau^{\prime}$ and $\alpha^{\prime}\in(0,1]$ are unrelated to those of Eqs. \eqref{23/05_2}. Moreover, none of the formulae listed in Eq. \eqref{23/05_2} are the Laplace transform of Eq. \eqref{23/05_3}. Efforts to shed more light on this troublesome situation have been the subject of both experimental and theoretical research accompanying the issues of dielectric physics \cite{Alvarez91,Alvarez93,KGorska21c,HandH96,RHilfer02,RHilfer02a}.   

Focusing attention on the HN pattern we remind that it was introduced in \cite{SHavriliak67} to explain the asymmetry of the Cole-Cole plot, i.e., the Argand-like diagram in the $\widehat{\varepsilon'}(\omega) = \RE[{\widehat{\varepsilon^{\star}}(\omega)}]$, $\widehat{\varepsilon''}(\omega) = \IM[{\widehat{\varepsilon^{\star}}(\omega)}]$ plane, observed in dielectric relaxation of polymers. Shortly after Havriliak and Negami's discovery, the HN pattern gained popularity because, although involving only three parameters, it appeared to be so flexible and universal that within a few years it became a standard model used to describe data obtained for a large range of relaxation and viscoelastic phenomena. Consequently, the calculation of its time domain counterpart, leading  either to the response or to the relaxation function, became a serious challenger to routine  parameterizations of their time evolution, based on the  KWW functions or finite sums of exponential decays.  This has motivated many authors to investigate the time dependence of the HN model and explore it using different mathematical tools. They include the analytical methods like the integral transforms theory and fractional calculus as well as an the approach worked out in the probability and stochastic processes theory and known  as subordination methods. We shall show that these seemingly distant approaches merge entirely into the framework of memory-dependent kinetics.  

The general expressions for the time-dependent counterparts of Eqs. \eqref{23/05_2} (as well as the frequency-dependent analogue of Eq. \eqref{23/05_3}) were found analytically  20 years ago by  R. Hilfer \cite{RHilfer02, RHilfer02a} who, for arbitrary real values of $\alpha$ and $\beta$, expressed { functions relevant for dielectric relaxation in terms of the Fox H-functions. Note that the applicability of this class of functions was used 10 years ealier for studies of other relaxation phenomena, especially in the research devoted to rheological models of viscoelastic materials \cite{Gloeckle91,Gloeckle93,RMetzler95,RMetzler03,TNonnenmacher91,HSchiessel95,NWTschoegl}}. In applications essential for the relaxation phenomena, the Fox H-functions are replaced by much better-known functions  of the Mittag-Leffler family \cite{RGarrappa16} and the Meijer G-function. In fact, its well-known representative, namely the standard Mittag-Leffler function, naturally enters the time relaxation for the CC model.  Following and exploring this direction of theoretical investigations we shall show how to express the functions relevant to the HN and JWS relaxation models, like the response or the relaxation functions and the probability densities, in terms of the Meijer G-function and the generalized Mittag-Leffler functions. These families of functions are well-implemented in modern computer algebra systems (CAS) and may be effectively studied using methods of classical analysis much more accessible to a larger community. The crucial step to achieve this purpose is to replace the irrational values of the parameters sitting in the Fox H-functions with rational numbers. This simple, but physically well-justified \cite{F3} approximation, enables us, as the first step, to represent results in terms of the Meijer G-functions, and in the second step in terms of the generalized hypergeometric functions which specialize to the generalized, or multi-parameter, Mittag-Leffler functions.

\section{The time domain description: general properties of the relaxation functions}\label{sec2}

Assume that the object of our investigation is the time-dependent quantity $\mathcal{J}(t)$, like the relaxation or the response function is,  which reflects the time behavior of some, suitably chosen, property $\mathcal{P}(t)$ of the system. For example, $\mathcal{P}(t)$ can be the number of its constituents which at an instant of time $t$ are still excited (have survived the decay), or the depolarization current, or the intensity of luminescence light, etc. The normalized form of ${\mathcal J}(t/\tau)$, with the characteristic time scale $\tau$ put in, reads:
\begin{equation}\label{1/05}
{\mathcal J}(t/\tau)=\frac{\mathcal{P}(t/\tau)-\mathcal{P}_{\infty}}{\mathcal{P}_{0}-\mathcal{P}_{\infty}}
\end{equation}
where $\mathcal{P}_{0}$ and $\mathcal{P}_{\infty}$ denote the values of $\mathcal{P}(t)$ for $t=0$ and $t=\infty$, respectively. Parameterization given by Eq. \eqref{1/05} implies by construction that ${\mathcal J}(0)=1$ and
 ${\mathcal J}(\infty)=0$. We also assume that $\mathcal{J}(t)$ is experimentally measurable (at least in the sense of {\em Gedankenexperiment}) and that in the system, there exist no internal factors which can revert the decay, i.e., $\mathcal{J}(t)$ never increases. Such requirements are too vague to specify $\mathcal{J}(t/\tau)$ satisfactorily and additional assumptions are needed to make any prospective mathematical modeling of $\mathcal{J}(t/\tau)$. Two commonly used suggestions are:
\begin{proposition}
{\rm
\begin{equation}\label{2/05}
{\mathcal J}(t/\tau) = \exp\left(\!-\int_{0}^{t}r(t^{\prime}/\tau)\D t^{\prime}\!\right),
\end{equation}
with
\begin{equation}\label{3/05}
r(t/\tau)=-\tau\frac{\D\, \ln {\mathcal J}(\xi)}{\D\xi}\vert_{\xi=t/\tau},
\end{equation}
which generalizes the standard exponential Debye decay law, $\D n(t)/\D t=-\lambda n(t)$, by extending it to the case when the relaxation/decay rate $\lambda$ becomes the time-dependent quantity, $\lambda\rightarrow r(t)$ instead of being a constant.}
\end{proposition}
\begin{proposition}
{\rm 
 \begin{align}\label{4/05}
 \begin{split}
 {\mathcal J}(t/\tau) & = \int_{0}^{\infty} p(x)\exp{(-xt/\tau)}\D x \\
 & = \tau\int_{0}^{\infty} p(\tau\xi)\exp{(-\xi t)}\D\xi 
 \end{split}
 \end{align}
 is interpreted as a continuous weighted  sum of the Debye decays, provided the additional condition that $p(x)$ is the probability distribution function (PDF), i.e., it is non-negative and normalizable. If we focus attention on the relaxation phenomena and identify ${\mathcal J}(t/\tau)$ either with the relaxation $n(t)$, or the response function $\phi(t)$, then both these functions may be represented as the Laplace integrals \cite{F4}
 \begin{equation}\label{18/06-10}
 n(t) = \int_{0}^{\infty} \E^{- \frac{t}{\tau} \xi} g(\xi) \D\xi
 \end{equation}
 and
 \begin{equation}\label{18/06-11}
 \phi(t) = -\frac{\D n(t)}{\D t} = \int_{0}^{\infty}  \E^{- \frac{t}{\tau} \xi}\, \frac{\xi}{\tau}\, g(\xi) \D\xi.
 \end{equation}
 }
 \end{proposition}
\medskip
\noindent 
If we have enough information on the time dependence of $n(t)$ or $\phi(t)$, e.g., know that their time decay is governed by some KWW-like function, then Eqs. \eqref{18/06-10} and \eqref{18/06-11} may reduce to the well-defined problems of classical mathematical analysis whose solutions allow us to determine the function $g(\xi)$, see e.g. \cite{HPollard46}. 
 
The above propositions are strongly supported by simple physical intuition and, at first glance, look very attractive as a way to solve the problem. Indeed, imagine that the constituents of the relaxing system are under consideration, e.g., polarized dipoles, decay according to the Debye law but each constituent of the system does that with a different decay rate $\lambda_{k}$; we treat the latter quantity as a non-negative discrete or continuous random variable whose distribution $g(\lambda)$ we somewhat know. We also assume that the randomization of elementary decays does not influence the Debye pattern dominating the process. Under such formulated conditions of independence  the use of a joint probability formula for the fraction of constituents which survive the decay gives 
\begin{equation}\label{7/02-5}
\frac{n(t)}{n_{0}} = \int_{0}^{\infty}e^{-\lambda  t } g(\lambda) \D\,\lambda, 
\end{equation}
i.e., we have arrived at an expression which has an obtrusive interpretation to be treated as a weighted sum of exponential decays. Analogous reasoning is the basis of the majority of routine analyses whose goal is to represent decay curves describing experimental data as a weighted finite (or infinite) sum of exponentials \cite{RSAnderssen11, Bertelsen74, Feldman02, Johnston06, Nigmatullin12}. However, despite their attractive simplicity and possible utility, the statements cited above are too naive. Equations \eqref{3/05} and \eqref{4/05} suffer from essential physical shortcomings and lack of mathematical precision. Equation \eqref{3/05}, if treated as a source of experimental information concerning $r(t)$, is difficult to be verified both for short and long times. Another objection is that $r(t)$ calculated from exactly solvable fractional models of relaxation becomes singular for $t\to 0$ and thus unphysical, e.g.,  \cite{KGorska21CNSNCa}. The main criticism against Eq. \eqref{4/05} comes from the objection that the interpretation of factors $\exp{(-\xi t)}$ as the Debye decays, each characterized by its own decay rate $\xi$, is formal and difficult to be explained. Does it mean that decay rates appearing in such a way are distributed according to some mysterious function $p(\tau\xi)$? What is the origin of the time scale $\tau$ introduced \textit{ad hoc} because of dimensional reasons but finally acquiring universal meaning? Nevertheless, so far proposed Eqs. \eqref{3/05} and \eqref{4/05}, if understood in mathematically and physically correct ways, provide us with useful guideposts which show how to choose, arrange and push forward theoretical methods whose development will lead to a better understanding of relaxation phenomena. 
 
The mathematical condition which guarantees that Eq. \eqref{4/05} makes sense, comes from classical mathematical analysis - according to the Bernstein theorem (for a brief tutorial on it, on completely monotone (CM) and on Bernstein (B) functions see Sect. \ref{subsec3.3}) the sufficient and necessary condition to represent the function $f(t)$ through the Laplace integral of the measure $g(x)$
\begin{equation}\label{5/05}
f(t)=\int_{0}^{\infty} \exp{(-xt)}g(x)\D x
\end{equation}
is that $f(t)$ is CM. In Eq. \eqref{5/05} we deal with $t\in\mathbb{R}_{+}$. It means that we work in the realm of real-valued Laplace integral whose continuation to the complex-valued Laplace transform needs care. Merging Eqs. \eqref{3/05} and \eqref{4/05} and requiring complete monotonicity of ${\mathcal J}(t/\tau)$ represented by Eq. \eqref{2/05}, one learns that the minus exponent, which is in its right-hand side (RHS), i.e., $\int_{0}^{t}r(t/\tau)\D t$, must belong to the class of Bernstein functions,  $\int_{0}^{t}r(t/\tau)\D t\in{\rm{B}}$. Thus, $r(t)$ must be the CM and, consequently, must not be singular. As mentioned earlier, solvable models, e.g., the CC, show that it is not the case, which excludes them from further considerations. The distinguished role of the CM functions in the theory of relaxation phenomena has been noted a long time ago, starting with early observations concerning viscoelasticity \cite{Day70} and rheology \cite{Bagley83}. Additional links between the CM properties and relaxation phenomena have been provided by the studies of Green's functions \cite{Anh03} and of the extended family of the Mittag-Leffler-type functions \cite{ECapelas11, RGarrappa16, KGorska21b, Hanyga1, FMainardi15, Tomovski14}, describing of the non-Debye relaxation.

Besides their well-established place in the classical mathematical analysis \cite{Berg}, the CM and B functions are objects which play an important role also in the probability and the theory of stochastic processes \cite{RLSchilling10}. Thus, their appearance in the formalism signalizes existence of non-accidental relations between probabilistic tools and analytical methods of fractional and operational calculi, if applied to the relaxation theory. Requiring probabilistic interpretation of Eq. \eqref{4/05}, besides of the indication that ${\mathcal J}(t)$ should be CM, leads to another important result. If we calculate the  Laplace transform of ${\mathcal J}(t)$ with respect to $t$, $\widehat{{\mathcal J}}(z)={\mathscr L}[{\mathcal J}(t),z]$, $z\in\mathbb{C}$, then, after changing the order of integration, we get
\begin{equation}\label{6/05}
 \widehat{{\cal J}}(z) = \int_{0}^{\infty}\frac{p(t)}{t+z}\D t.
 \end{equation}
Under the condition that $p(t)$ is integrable and non-negative for all $t\in{\mathbb{R}_{+}}$, it means that $\widehat{{\mathcal J}}(z)$ belongs to the class of Stieltjes functions (S) which obeys well-known analytical properties, play important role in the Laplace transform theory \cite{Widder} and, as we will see, turns out to be of a crucial importance in our further analysis. Notice that to get Eq. \eqref{6/05}, we assume $\RE{(z)}\in{\mathbb{R}}_{+}$. From the physical point of view, this assumption is too restrictive - to make our scheme compatible with physical data incorporated in Eq. \eqref{23/05_2}, we do need to include also $\RE{(z)}=0$. The justification of this statement comes from the fact that the Laplace transform Eq. \eqref{6/05} for $z=\I\!\omega$ may be considered as  the Fourier transform of a function vanishing on the negative semiaxis, i.e.,  the Fourier transform of a function $\Theta(t)f(t)$ being undefined for $t=0$. The fundamental physical implication of the Fourier transform, namely the time-frequency correspondence, implies that $\widehat{{\mathcal J}}(z)\vert_{z=\I\!\omega}$  in such a way understood  is closely related to the earlier introduced spectral functions being objects of direct physical meaning, as  reconstructed from experimental data provided by methods of the broadband dielectric spectroscopy. But the spectral functions should be (and they are) defined also for $\omega=0$, in the case of a static field. Mathematically, it means that we have to point out the value $\widehat{{\mathcal J}}(0)$ to redefine the singularity of $\int_{0}^{\infty}\frac{p(t)}{t}\D t$. Here the Plemelj-Sokhotski formula \cite{F5} enters the game
 \begin{equation}\nonumber
 \int_{0}^{\infty}\frac{p(t)}{t}\D t = \I\pi p(0) +{ \rm{PV}}\int_{0}^{\infty}\frac{p(t)}{t}\D t
 \end{equation}
 with the Cauchy principal value 
 \begin{equation}\nonumber
{\rm{PV}}\int_{0}^{\infty}\frac{p(t)}{t}\D t = \lim_{\epsilon\to 0}\int_{\epsilon}^{\infty}\frac{p(t)}{t}\D t 
 \end{equation}
  being introduced. Consequently, we should extend Eq. \eqref{6/05} to the form 
 \begin{equation}\nonumber
 \widehat{{\mathcal J}}(z)={\rm{const}} + {\rm{PV}}\int_{0}^{\infty}\frac{p(t)}{t+z}\D t,
 \end{equation}
 which explains different shapes of the spectral functions relevant for the HN and JWS models.
   
 Summarizing the results obtained so far  -  using intuitive, and from time to time also hand-waving arguments, we can expect that functions admittable to describe non-exponential relaxation phenomena should be searched among functions belonging to the CM class or eventually its subclasses. Natural justification of such a conjecture is that the CM functions obey characteristic properties of relaxation functions, namely they are always non-negative and non-increasing. Infinite differentiability and sign-alternating derivatives, which in fact define complete monotonicity, are impossible to be reliably verified for functions known from phenomenological analyses only. Intuitively, to require the CM property shared with the exponential function used to describe the standard Debye decay, can shed more light on the problem.  In Secs. \ref{sec5} and \ref{sec6} we will see that it is really the case. Now, leaving aside the CM functions, we are going to convince the readers that the most promising factor which will enable us to go forward is to consider memory-dependent kinetics giving dynamical laws governing the relaxation phenomena.                     
 
  \section{Memory dependent kinetics}\label{SEC3}

As signalized in the Introduction, the characteristic feature of the relaxation phenomena is that the relaxing system exhibits delayed response to the changes of external conditions. A typical example of such a behavior is extended over time approach to equilibrium after switching off factors perturbing the system during its previous history. In other words -  approaching equilibrium, even in the case of absent external influence, needs time and the process under consideration does not satisfy the instantaneous response principle. This means giving up the time-local evolution rules and strongly suggests looking for prospective challengers which would be used to replace the standard time-local evolution equations. The ``first choice'' solution of the problem is to take into account the time non-local analogues of basic equations, in particular to consider integro-differential equations with kernels suited to mimic memory effects occurring in the system. Thus, to study the non-Debye relaxation phenomena, as well as other kinetic phenomena taking place in complex systems, e.g., anomalous diffusion, we should to reformulate the standard time-local approach. The aim is to construct a new theoretical scheme which has the time non-locality built-in from the very beginning, and leads to new evolution laws modified by various memory effects. The cornerstone of such new formalism is classical works of R. Zwanzig \cite{Zwanzig1,Zwanzig2} who, after paying attention to the fact that the formalism of instantaneous response did not work in many physically interesting situations, proposed the self-consistent construction of kinetic equations respecting causality and taking into account delayed response of the system. The essence of mathematical encoding such delays is to replace point-like operations, e.g., the usual multiplication of functions, by integral operators taken in the form of time convolutions of memory functions and solutions looked for. Thus, it is justified to say that such modified evolution equations are the "time-smeared"  standard ones. Almost 40 years later, at the very beginning of the current century, it was I. M. Sokolov \cite{Sokolov1} who shed new light on the ideas underpinning Zwanzig's seminal works.  Sokolov showed that the use of generalized non-Markovian Fokker-Planck (FP) equations with built-in memory kernels leads to considerable progress in understanding anomalous transport and non-Debye relaxations. Subsequent years brought new research techniques and huge amount of earlier unreachable experimental data. These, urgently needed to be theoretically explained, pushed forward theoreticians' efforts  and soon led to the development of effective mathematical tools rooted either in mathematical analysis, e.g.,  provided by the fractional calculus, or coming from the probability/stochastic processes theory. Taken separately or merged together, these methods have formed a toolbox whose usage has enabled physicists to solve equations of anomalous diffusion and non-Debye relaxations for a quite large number of problems  \cite{ AChechkin21, KGorska21, TSandev18, TSandev17, AStanislavsky15}. 

\subsection{The basics}\label{sub3.1}

Restricting Zwanzig's approach to the problems depending on the time only, e.g., to the simplest analysis  of dielectric relaxation \cite{F6}, but holding on to the spirit of seminal paper \cite{Zwanzig1} we can write down the evolution equation as
\begin{equation}\label{14/10-1}
\frac{\D n(t) }{\D t}= {\mathcal{O}}n(t),
\end{equation}
where $t\in\mathbb{\mathbb{R}_{+}}$. The linear operator ${\mathcal{O}}$ acts on the time variable but is not simple multiplication by a time-dependent function, and is the object which represents a more general characterization of non-locality in time. If ${\mathcal{O}}$ is represented as an integro-differential operator with the kernel $O(t-\xi)\mathcal{B}$, $\mathcal{B}={\rm const}$, then Eq. \eqref{14/10-1} becomes
\begin{equation}\label{14/10-2}
\frac{\D n(t) }{\D t} = \int_{0}^{t}\!\! O(t-\xi) \mathcal{B} n(\xi) \D\xi, 
\end{equation}    
i.e., has a typical form of the Volterra equation. It is convenient to rewrite Eq. \eqref{14/10-2} as  
\begin{equation}\label{13/09-3}
\frac{\D n(t) }{\D t} = -\frac{\D}{\D t}\!\int_{0}^{t}\!\! M(t-\xi) \mathcal{B} n(\xi) \D\xi 
\end{equation}
coming directly from the integral master equation \cite{AStanislavsky17,AStanislavsky19b}
\begin{equation}\label{14/10-4}
n(t) = n(0) - \int_{0}^{t} M(t - \xi) \mathcal{B} n(\xi) \D\xi.
\end{equation}
If we take the Laplace transforms (with respect to $t$) of Eqs. \eqref{14/10-2} and \eqref{13/09-3} or \eqref{14/10-4} and symbolically write down these equations as 
\begin{equation}\nonumber
[z - \widehat{O}(z) \mathcal{B}]\, \widehat{n}(z) = z [1 - \widehat{M}(z)\mathcal{B}]\, \widehat{n}(z) = n(0)
\end{equation}
then we see that  Eqs. \eqref{14/10-2} - \eqref{14/10-4} are equivalent, if in the Laplace domain $\widehat{O}(z) = z\widehat{M}(z)$, i.e., in the time domain $O(t) = \D M(t)/\D t + M(0+)$, where $M(0+)=\lim_{t\rightarrow 0_{+}}M(t)$ is to be specified, cf. \cite{AChechkin21}.
Looking for solutions of \eqref{14/10-2} - \eqref{14/10-4} we are interested in the functions $n(t)$ which are normalizable and non-negative, i.e., fulfill a minimal set of requirements needed to interpret them as the relaxation  functions. The crucial property to be shown is non-negativity - we will see that requiring this property allows to find conditions which have to be put on the memory functions.    
   
The theory of integral equations, more precisely the possibility of a dichotomic description of the system using either integral or differential equations, suggests that it would be interesting also for Eq. \eqref{14/10-4} to find its integro-differential analogue. To achieve this goal let us notice that if we have given Eq. \eqref{13/09-3} governed by the kernel function $M(t)$ then we may ask for a function $k(t)$ defined by the convolution 
\begin{equation}\label{14/10-5}
\int_{0}^{t}k(t - t^{\prime})M(t^{\prime})\D t^{\prime} = 1.
\end{equation}
which, if required to be satisfied for any $t > 0$, is known in the theory of integral equations as the Sonine relation \cite{AHanyga20, AHanyga21, YLuchko21, YLuchko21a, Meerschaert19a}. Transformed to the Laplace domain Eq. \eqref{14/10-5} becomes $\widehat{k}(z)\widehat{M}(z) = z^{-1}$ from which it follows:
\begin{equation}\label{28/09-1}
\widehat{k}(z)= [z \widehat{M}(z)]^{-1}. 
\end{equation}
General properties required from $M(t)$ in order to satisfy basic conditions of the Laplace transform theory guarantee that Eq. \eqref{14/10-5} defines the couple $(M, k)$ uniquely. Moreover, the symmetry of Eq. \eqref{28/09-1} with respect to the interchange $M(t)\leftrightarrow k(t)$  means that if $M(t)$ is interpreted as a memory function then $k(t)$ does share this property: if {{Eq. \eqref{14/10-5}}} holds then any suitable $M(t)$ which governs {{Eq. \eqref{13/09-3}}} has its partner memory function $k(t)$ entering the kernel of integro-differential equation    
\begin{equation}\label{13/09-1a}
\int_{0}^{t} k(t-\xi) \frac{\D {n}(\xi)}{\D\xi} \D\xi = - \mathcal{B} n(t), \qquad \mathcal{B} = {\rm const}.
\end{equation}
Equation \eqref{13/09-1a} is equivalent to {{Eq. \eqref{13/09-3}}} and in mathematics it is known as the Caputo-Djhrbashyan problem \cite{AKochubei11}. The equivalence of Eqs. \eqref{13/09-3} and \eqref{13/09-1a} is seen if we calculate the time derivative in the RHS of Eq. \eqref{13/09-3} and next integrate the resulting equation with respect to $\int^{t}_{0}\D t^{\prime}$ using the Sonine relation \eqref{14/10-5} twice \cite{KGorska21a, KGorska20a}. Duality in the description of the same process using either Eq.  \eqref{13/09-3} or Eq. \eqref{13/09-1a} was noticed and studied in many papers \cite{KGorska20a, KGorska21a, TSandev17, TSandev18}. Those were first of all addressed to the anomalous diffusion (through studying various applications of the generalized Fokker-Planck equation \cite{TSandev17, TSandev18}) but led to results by no means limited to this class of phenomena \cite{KGorska20a, KGorska21a}. To name different points of view on the problem we list: i.) solvability of the Cauchy problem for Eq. \eqref{13/09-1a} studied either using advanced partial differential equations methods  \cite{AKochubei11} or operator calculus tools developed for the Volterra-type evolution equations \cite{JPruess93} and leading to the appearance of subordination \cite{Bazhlekova18,Bazhlekova19},  ii.) investigations devoted to the role played by the Sonine relation and resulting duality of memories \cite{KGorska20a, KGorska21CNSNCa, KGorska21a, KGorska21c, AHanyga20, AStanislavsky21} and iii.) efforts oriented on understanding mutual relations between the so-called deterministic and probability/stochastic approaches to anomalous kinetic phenomena \cite{KGorska21, TSandev18, TSandev17, AStanislavsky17, AStanislavsky19b}.

Recall that restricting the kinetic problem under consideration to be depending on the time only means that the action of the FP operator reduces to multiplication by a constant factor $\mathcal{B}$. Thus we arrive at equations governing the relaxation phenomena which provide us with other examples of non-Markovian processes \cite{AStanislavsky19b, AStanislavsky15}. It has to be noted that Eq. \eqref{14/10-4} enables us to write down the spectral function, i.e., the Laplace transform $\widehat{\phi}(z)$ of the response function $\phi(t)= -\D{n}(t)/\D t$ in terms of $\widehat{M}(z)$:
\begin{equation}\label{21/09-1}
\widehat{\phi}(z) = 1 - z \widehat{n}(z) = \{1 + [{\mathcal B} \widehat{M}(z)]^{-1}\}^{-1},
\end{equation}
reasulting in
\begin{equation}\label{21/09-1a}
\widehat{M}(z) = {\mathcal B}^{-1}\frac{\widehat{\phi}(z)}{1 - \widehat{\phi}(z)},
\end{equation}
which gives the mutual relation between the spectral functions and evolution kernels responsible for the memory effects. This observation has far-going consequences. From the physicists point of view Eqs. \eqref{21/09-1} and \eqref{21/09-1a}, if taken for $z = \I\omega$ with $\omega$ identified as the frequency of harmonic field used to polarize the sample, couple in a unique way evolution equations introduced on theoretical background with phenomenologically known spectral functions given by formulae which classify the standard non-Debye relaxations as the CC, CD, HN, and JWS patterns (cf. \cite{KGorska21CNSNCa, KGorska21a, KGorska21c, KGorska20a}).   

\subsection{The stochastic approach - a few general remarks } 

As mentioned, Eqs. \eqref{13/09-3}, \eqref{14/10-4}, and \eqref{13/09-1a} are the Volterra-type equations \cite{GGripenberg90}. They may be considered as equations which, because of a possible introduction of "atypical" kernels  $M(t)$ and $k(t)$, go beyond the fractional differential equations conventionally used to describe the anomalous diffusion and non-Debye relaxation phenomena. In turn, Eq. \eqref{14/10-4} may be identified as the master equation which governs some stochastic processes underlying mechanisms of anomalous kinetics \cite{AStanislavsky19b}. Recall that the memory functions $M(t)$ and $k(t)$ play a dual role in our scheme - they determine the mathematical structure of equations under study and provide us with links to observational data fitted in the case of relaxation experiments by phenomenological spectral functions $\widehat{\phi}(\I\!\omega)$. However, if we restrict ourselves to the experimental data only, then usually we are not able to acquire knowledge sufficient to determine basic processes underlying physics of considered phenomena.  To proceed further  we do need more detailed information, also coming from mathematics, expected to be obtained from analysis of properties which the memory functions should obey. Obvious requirements that the memory functions have to be non-negative and non-increasing are insufficient to judge the existence and physical applicability of solutions as well as to find their interpretation. Some extra conditions, like the Boltzmann fading memory concept, did not satisfy initial expectations due to the lack of mathematical precision \cite{RSAnderssen02,RSAnderssen02a}. Clarification and resolution of many doubts appearing in the statistical description of both anomalous transport and relaxation phenomena came with using advanced probabilistic methodology, first of all with identification of just mentioned phenomena with stable stochastic processes. 

Stochastic point of view (cf. e.g., \cite{Meerschaert01,Meerschaert19,AStanislavsky17, AStanislavsky15} and \cite[Chs. 4.1, 4.3]{AStanislavsky19b}) led to methods which made it possible to analyse non-Debye relaxations (as well as other anomalous kinetics processes) using  probabilistic concepts like subordination { \cite{Bochner1, Bochner, Feller, RLSchilling98}}, infinitely divisible distributions \cite{Lukac}  and applications the theory of B and S functions, see \cite{Berg,RLSchilling10} and Sect. \ref{subsec3.3} for a brief tutorial. Simultaneously, investigations based on less popular methods of mathematical analysis (fractional and operational calculi) made the integro-differential equations \eqref{13/09-3} and \eqref{13/09-1a} much better understood if they were governed by kernels whose the Laplace images belong to the S class \cite{KGorska21,AKochubei11, AKochubei19}. Here we would like to turn once more the readers attention to the importance of the Sonine relation and the special role of the S functions hidden behind it. Indeed, if one requires that $\widehat{M}(z)$ and $\widehat{k}(z)$ of Eq. \eqref{28/09-1} belong to the same class of functions then the most natural and convincing 
possibility to satisfy this requirement is to put both of them in the S class. Additionally, belonging to the S class directly links the memory functions with the complete Bernstein (CB) functions and the Laplace exponents, objects which characterize PDFs relevant for infinitely divisible stochastic processes { \cite{KGorska21CNSNCa, KGorska21a, KGorska21c, RLSchilling98, AStanislavsky19b}}.
   
\subsection{Integral decompositions and subordination}\label{SEC3.1}

To keep the forthcoming  construction as general as possible let us observe that for the relaxation phenomena $\widehat{n}(z)$ from 
Eq. \eqref{21/09-1} may be rewritten in the form
\begin{equation}\label{25/01/23-1}
\widehat{n}(z) = \frac{[z \widehat{M}(z)]^{-1}}{\mathcal{B} + [\widehat{M}(z)]^{-1}} = \int_{0}^{\infty} \E^{- \xi\mathcal{B}}\, \frac{\widehat{\Psi}(z)}{z} \E^{- \xi \widehat{\Psi}(z) } \D\xi.
\end{equation}
To get Eq. \eqref{25/01/23-1} the Schwinger parametrization was used and $ \widehat{\Psi}(z) = [\widehat{M}(z)]^{-1}$ introduced. $\widehat{\Psi}(z)$ is called the characteristic or the Laplace-L\'{e}vy exponent \cite{AStanislavsky17, AStanislavsky19b} which can be connected to the spectral function $\widehat{\phi}(z)$ via Eq. \eqref{21/09-1a} : 
\begin{equation}\label{23/02/23-1}
\widehat{\Psi}(z)= \mathcal{B}\,\frac{1 - \widehat{\phi}(z)}{\widehat{\phi}(z)}.
\end{equation}
Taking the inverse Laplace transform of Eq. \eqref{25/01/23-1} we express it in the form of integral decomposition:
\begin{equation}\label{30/07-1}
n(t) = \int_{0}^{\infty} n_{D}(\xi) f_{\widehat{\Psi}}(\xi, t) \D\xi
\end{equation}
the integrand being a product of the Debye relaxation function $n_{D}(\xi) = \exp(-\mathcal{B}\xi)$ and $f_{\widehat{\Psi}}(\xi, t) = \mathscr{L}^{-1}[\widehat{\Psi}(z) \exp[-\xi \widehat{\Psi}(z)]/z; t]$ presenting the relation between the time $t$ and integral variable $\xi$. The functions $n_{D}(\xi)$ and $f_{\widehat{\Psi}}(\xi, t)$ are given by normalized and non-negative functions which in probability theory means that they are the PDFs of ``parent'' \cite{F7} and ``leading'' processes, respectively. If these PDFs are independent Eq. \eqref{30/07-1} turns out to be the subordination. Within the latter scheme the integral variable $\xi$ is interpreted as the internal time with randomize the physical time $t$. 

 The same results can be obtained by using the Efros theorem (Subsect. \ref{subsec3.1}) giving more flexibility than an application of the Schwinger parametrization. The use of the Efros theorem allows one to preset $n(t)$ as
\begin{equation}\label{26/01/23-1}
n(t) = \int_{0}^{\infty} h(u) f(u, t) \D u,
\end{equation}
where $h(u)$ denotes the relaxation function of basic model. It is subordinated by the PDF $f(u, t)$ which expresses the conversion between physical $t$ and internal $u$ times. The simplest realization of Eq. \eqref{26/01/23-1} is for $f(u, t) = \delta(t-u)$ and $h(u) = n(u)$ where $\delta(\cdot)$ denotes the $\delta$-Dirac distribution. Other examples of Eq. \eqref{26/01/23-1} are presented in Subsecs. \ref{sec5.6} and \ref{sec6.6} where for the HN and JWS models setting separately $h(u)$ and $f(u, t)$ we obtain dual representations of the same $n(t)$.

 \section{Interlude  - elements of the mathematical toolbox}\label{sec3}
 
The wide applicability of the relaxation pattern attracts numerous authors to investigate it from various points of view. The goal of mathematically oriented research was to find the analytic expressions for the family of relaxations in the time domain. We mentioned in the Introduction that the problem was solved by R. Hilfer \cite{RHilfer02, RHilfer02a} who showed that solutions looked for may be expressed in terms of the Fox H-functions. This result, important for theoretical considerations, regrettably did not find deeper practical interest among experimentalists. It is not difficult to understand that as the Fox H-functions are special functions which for general values of parameters are defined symbolically either via the contour integrals of the Mellin-Barnes type \cite{APPrudnikov-v3} or by slowly convergent series.  All together it makes relevant calculations difficult to understand, time-consuming and prone to mistakes emerging during hand-made manipulations. Thus, it is justified to emphasize that the use of the Fox H-functions formalism is not mandatory to invert analytically the Laplace transform in Eq. \eqref{23/05_2} if its RHS specializes in specifically chosen sets of parameters. As is shown in a series of research papers and books, e.g., \cite{Prabhakar, book1} in such a case the suitable Laplace transform essentially simplifies and results in functions of the Mittag-Leffler family, much better known and easier to be handled both analytically and numerically. Moreover, relevant Mittag-Leffler functions obey properties placing them in special classes of analytic functions which for the positive value of the argument, besides of being non-negative, also boils down to the completely monotone functions \cite{Berg, RLSchilling10}.

\subsection{Non-negatively definite functions}\label{subsec3.3} 

\subsubsection*{The completely monotone (CM) functions}\, are introduced as non-negative and infinitely differentiable  (i.e., belonging to the $C^{\infty}$ class) functions $\widehat{g}(s)$ on $\mathbb{R}_{+}$ whose all derivatives alternate in sign for any $s\in\mathbb{R}_{+}$:  
\begin{equation}\nonumber
(-1)^{n}{\widehat{g}}^{(n)}(s)\ge 0, \quad n = 0, 1, \ldots.
\end{equation}
Notice that any CM function is a non-increasing and convex function.\\
According to the {\bf Bernstein theorem} \cite{RLSchilling10}, we can connect in a unique way  CM and non-negative functions: $s\in [0,\infty)\rightarrow \widehat{g}(s) \in \textrm{CM}$ iff  
\begin{equation}\nonumber
\widehat{g}(s) = \int_{0}^{\infty} \exp(-s u) g(u)\!\D u 
\end{equation} 
with $g(u)\ge 0$ for all $u\in [0,\infty)$.  
The important property of the CMs is that the product of two CM functions is CM as well, i.e., $CM \cdot CM \subset CM$.

\subsubsection*{The Stieltjes (S) functions} are denoted as $\widehat{k}(s)$ and form a subclass of the CM functions. On the real semiaxis they admit the integral representation 
 \begin{equation}\label{11/06-10}
 \widehat{k}(s) = \frac{a}{s} + b + \int_{0}^{\infty} \frac{\sigma(\D u)}{s + u}, \qquad s\in\mathbb{R}_{+},
 \end{equation}
 where $a, b \geq 0$ and $\sigma$ is a Borel measure on $[0, \infty)$, such that $\int_{0}^{\infty} (1 + u)^{-1} \sigma(\D u) < \infty$.  
 Equation \eqref{11/06-10} can be interpreted as the Laplace transform of the CM function, {i.e.,} the Laplace transform of the Laplace transform of the non-negative function \cite{Berg, RLSchilling10}. This fact was pointed out and commented in the Introduction. 

\subsubsection*{The Bernstein (B) functions} are these non-negative infinitely differentiable functions $\widehat{b}(s)$ on $\mathbb{R}_{+}$  whose first derivative is CM. It means that
\begin{equation}\nonumber
(-1)^{n-1}\, \widehat{b}^{(n)}(s) \geq 0, \quad n = 1, 2, \ldots.
\end{equation}
All B functions are non-decreasing and concave. As the main property, we remark that the composition of CM function with B function gives another CM function, i.e., $CM(B)\subset CM$. This property is illustrated on the example of $\exp[-b(s)]$ which for $b(s)\in B$ is CM function.

\subsubsection*{The complete Bernstein (CB) functions,} denoted as $\widehat{c}(s)$, $s\in\mathbb{R}_{+}$, are CB functions such that $\widehat{c}(s)/s$ is the Laplace transform of the CM function restricted to the positive semi-axis, or, equivalently, the same way restricted Stieltjes transform of a positive function. \\
\noindent
In the paper we will use the following properties of CB functions:
\begin{itemize}
\item[(a1)] the composition of two CB functions is CB function, i.e., $CB(CB) \subset CB$; 
\item[(a2)] S functions can be obtained by making the composition of another S function with a CB function or the composition of a CB function with another S function, i.e., $S(CB) \subset S$ and $CB(S) \subset S$; 
\item[(a3)] reciprocal (algebraic inverse) of an S function is a CB function. The inverse property is true as well: $1/\widehat{k}(s)$ is CB function and $1/\widehat{c}(s)$ is a S function for $\widehat{k}\in S$ and $\widehat{c}\in CB$. 
\end{itemize}

The illustrative example of differences between CM, S, B, and CB functions is provided by the power law function. This is presented in Tab. \ref{tab1}.
\begin{table}[h]
\begin{center}
\begin{tabular}{c | c | c | c| c}
$f(s)$\,\,\, & \,\,\,CMF\,\,\, & \,\,\,SF\,\,\, & \,\,\,BF\,\,\, & \,\,\,CBF\,\,\, \\ \hline 
$s^{\mu}$ & $\mu \leq 0$ & $\mu \in[-1, 0]$ & $\mu\in[0, 1]$ & $\mu\in(0, 1)$ \\
$(s + b)^{\mu}$, $b > 0$ & $\mu \leq 0$ & $\mu \in[-1, 0]$ & $\mu\in[0, 1]$ & $\mu\in(0, 1)$ 
\end{tabular}
\end{center}
\caption{\label{tab1} Examples of CMF, SF, BF, and CBF \cite{RLSchilling10, TSandev18a}.}
\end{table}
      
 \subsection{Two theorems of classical mathematical analysis.}\label{subsec3.1} 
 
Two theorems of classical mathematical analysis are important for the mathematical toolbox used in the relaxation theory. First of them, which we abbreviate as Theorem 1, allows one to make analytic continuation of completely monotone functions to the  complex domain. Thus, if its conditions are satisfied, we can freely jump between the complex variable $z$, e.g., $z = \I\!\omega$ and real $s > 0$. The second one is the  Efros theorem \cite{ AApelblat21,Ditkin51,Efros35, KGorska12a, UGraf04, VSMartynenko68, LWlodarski52} which generalizes the Borel convolution theorem for the Laplace transform. We will see that if applied to our purposes, the Efros theorem justifies integral decompositions and can be interpreted as a guidepost leading to the subordination approach \cite{KGorska21,KGorska23}. Just mentioned theorems state as follows. 

\smallskip
\noindent
Characterization of the Laplace transform of a CM, locally integrable function (\cite{ECapelas11, GGripenberg90}) is given by {\bf Theorem 1}:\\ 
{\em The Laplace transform $\widehat{f}(z)$ of a function $f(s)$, $s>0$, that is locally integrable on $\mathbb{R}_{+}$ and completely monotone, has the following properties:
  \begin{itemize}
  \item[(a)] $\widehat{f}(z)$ has an analytical extension to the cut complex plane $\mathbb{C}\setminus\mathbb{R}_{-}$.
  \item[(b)] $\widehat{f}(s) = \widehat{f}^{\star}(s)$ for $s\in(0, \infty)$,
  \item[(c)] $\lim_{s\to\infty} \widehat{f}(s) = 0$,
  \item[(d)] $\IM[\widehat{f}(z)] < 0$ for $\IM(z) > 0$,
  \item[(e)] $\IM[z \widehat{f}(z)] \geq 0$ for $\IM(z) > 0$ and $\widehat{f}(s) \geq 0$ for $s\in(0, \infty)$.
  \end{itemize}
 Conversely, every function $\widehat{f}(z)$ that satisfies (a)--(c) together with (d) or (e), is the Laplace transform of a function $f(s)$, which is locally integrable on $\mathbb{R}_{+}$ and completely monotone on $(0, \infty)$}. 
 \vspace{-0.4\baselineskip}
\begin{proof}
The proof of this theorem is presented in Ref. \cite{GGripenberg90}.
\end{proof}
\vspace{-0.7\baselineskip}
\begin{remark}
{\rm 
The Laplace transforms of CM functions satisfying the assumptions of Theorem 1 are defined for $z\in \mathbb{C}\setminus\mathbb{R}_{-}$ and because of conditions (d) or (e) they can be called the Nevanlinna functions \cite{Akhiezier, Berg, RLSchilling10}.  For positive argument $s$ they are, according to {\rm{\cite{RLSchilling10}}}, the Stieltjes functions. 
}
\end{remark}
 
\smallskip
\noindent
The {\bf Efros theorem} \cite{Ditkin51,Efros35, UGraf04, VSMartynenko68, LWlodarski52}. \,  {\em If $\widehat{G}(z)$ and $\widehat{q}(z)$ are analytic functions, and 
\begin{equation}\label{30/03-0}
{\mathscr L}[h(\xi); z] = \widehat{h}(z)
\end{equation}
as well as
\begin{equation}\label{30/03-1}
{\mathscr L}[f(t); z] = \int_{0}^{\infty} f(t) \E^{-z t} \D t = \widehat{G}(z) \E^{-\xi \widehat{q}(z)},
\end{equation}
then 
\begin{equation}\label{30/03-1a}
\widehat{G}(z)\,\widehat{h}(\widehat{q}(z)) = \int_{0}^{\infty}\!\Big[\!\int_{0}^{\infty}\!\! h(\xi) f(\xi, t) \D\xi\Big] \E^{- z t} \D t.
\end{equation}}
\begin{proof}
The proof can be found in Ref. \cite{Efros35,VSMartynenko68}.
\end{proof}
\begin{remark}
{\rm 
From Efros theorem, it follows that 
\begin{equation}\label{30/03-4}
{\mathscr L}^{-1}[\widehat{G}(z)\widehat{h}(\widehat{q}(z)); t] = \int_{0}^{\infty}\!\! h(\xi) f(\xi, t) \D\xi.
\end{equation}
}
\end{remark}

 \subsection{Some special functions}\label{subsec3.2}
 
 \subsubsection*{The generalized hypergeometric functions} ${_{p}F_{q}}\Big({(c_{p}) \atop (d_{q})}; x\Big)$, $x\in\mathbb{R}$, are defined via the series \cite{NIST, APPrudnikov-v3}:
\begin{equation}\label{14/03-A5}
{_{p}F_{q}}\left({(c_{p}) \atop (d_{q})}; x\right) \okr \sum_{r=0}^{\infty} \frac{x^{r}}{r!} \frac{(c_{1})_{r} (c_{2})_{r}\cdots (c_{p})_{r}}{(d_{1})_{r} (d_{2})_{r}\cdots (d_{q})_{r}},
\end{equation}
where $(c)_{r} = \Gamma(c + r)/\Gamma(c) = c (c +1)\cdots (c+r-1)$ is the Pochhammer symbol (rising factorial). 

\subsubsection*{The Meijer G-functions.} The Meijer G-functions \cite{NIST, APPrudnikov-v3} are defined through the Mellin-Barnes integrals introduced as inverses in the sense of the Mellin transform of expressions being ratios of certain $\Gamma$ functions. The definition of the Meijer G-function reads
\begin{align}\label{2/08-1}
&G^{m, n}_{p, q} \left(x\Big\vert {(a_{p}) \atop (b_{q})}\right) \okr \\ \nonumber
&\int_{\gamma_{L}}  \frac{\prod_{i=1}^{m}\Gamma(b_{i} + s) \prod_{i=1}^{n}\Gamma(1 - a_{i} - s)}{\prod_{i = n+1}^{p}\Gamma(a_{i} + s) \prod_{i = m+1}^{q}\Gamma(1 - b_{i} - s)}\, x^{-s} \frac{\D s}{2\pi\!\I}.
\end{align}
Parameters in Eq. \eqref{2/08-1} are subject to conditions
\begin{align*}
\begin{split}
& x\neq 0, \quad 0 \leq m \leq q, \quad 0 \leq n \leq p; \\
& a_{i}\in\mathbb{C}, \quad i=1, \ldots, p; \qquad
 b_{i}\in\mathbb{C},  \quad i=1, \ldots, q; \\
& (a_{p}) = a_{1}, a_{2}, \ldots, a_{p}; \quad (b_{q}) = b_{1}, b_{2}, \ldots, b_{q}.
\end{split}
\end{align*}
The contour $\gamma_{L}$ lies between the poles of $\Gamma(1 - a_{i} - s)$ for $i=1,\ldots, n$ and the poles of $\Gamma(b_{i} + s)$ for $i=1, \ldots, m$. For generalization of Eq. \eqref{2/08-1} to the Fox H-function, see  {\cite[Eq. (8.3.2.21)]{APPrudnikov-v3} namely 
\begin{equation}\nonumber
G^{m, n}_{p, q} \left(z\Big\vert {(a_{p}) \atop (b_{q})}\right) = H^{m, n}_{p, q} \left[z\Big\vert {[a_{p}, 1] \atop [b_{q}, 1]}\right]. 
\end{equation}
Another possibility to express the Meijer G-function in terms of Fox H-function is provided by \cite[Eq. (8.3.2.22)]{APPrudnikov-v3}. It is obtained by using the Gauss--Legendre formula for the Gamma function in Eq. \eqref{2/08-1}:
\begin{equation}\label{15/06-2}
\prod_{j=0}^{n-1}\Gamma\left(z+\frac{j}{n}\right) = (2\pi)^{\frac{n-1}{2}} n^{\frac{1}{2}-nz} \Gamma(nz) ,
\end{equation}
for $z\neq 0, -1, -2, \ldots$, and $n=1, 2, \ldots$.}

\smallskip
\noindent
The main properties of Meijer G-functions are as follows: 
\begin{itemize}
\item[{(i)}] They form a set closed under the Laplace transform. It means that its (direct and/or inverse) Laplace transform leads to another Meijer G-function:
\begin{multline}\label{2/07-10}
\mathscr{L}\left[x^{\alpha-1} G^{m, n}_{p, q}\left(\omega x^{l}\Big\vert {(a_{p}) \atop (b_{q})}\right); z\right] = \\\frac{l^{\alpha - \frac{1}{2}} z^{-\alpha}}{(2\pi)^{\frac{l-1}{2}}} G^{m, n+l}_{p+l, q}\left(\frac{\omega l^{l}}{z^{l}}\Big\vert {\Delta(l, 1-\alpha), (a_{p}) \atop (b_{q})}\right), 
\end{multline}
and
\begin{multline}
\mathscr{L}^{-1}\left[z^{-\alpha} G^{m, n}_{p, q}\left(\frac{\omega}{z^{l}}\Big\vert {(a_{p}) \atop (b_{q})}\right); x \right] = \\ \frac{l^{1/2-\alpha} x^{\alpha-1}}{(2\pi)^{\frac{1-l}{2}}} G^{m, n}_{p, q+l}\left(\frac{\omega x^{l}}{l^{l}}\Big\vert {(a_{p}) \atop (b_{a}), \Delta(l, 1-\alpha)}\right). \label{2/07-11}
\end{multline}
The conditions on parameters under which Eqs. \eqref{2/07-10} and \eqref{2/07-11} are held are given in \cite[Eq. (2.24.3.1)]{APPrudnikov-v3} and \cite[Eq. (3.38.1.2)]{APPrudnikov-v5}, respectively. The symbol $\Delta(n, a)$ denotes the special list of parameters which reads $a/n, (a+1)/n, \ldots, (a+n-1)/n$. 

\item[{(ii)}] Below we quote two identities which will be used in our calculations:
\begin{align}\label{17/06-10}
G^{m, n}_{p, q}\left(z \Big\vert {(a_{p}) \atop (b_{q})} \right) & = G^{n, m}_{q, p}\left(\frac{1}{z} \Big\vert {(1 - b_{q}) \atop (1 - a_{p})} \right), \\ \label{17/06-11}
z^{\mu}\, G^{m, n}_{p, q}\left(z \Big\vert {(a_{p}) \atop (b_{q})} \right) & =  G^{m, n}_{p, q}\left(z \Big\vert {(a_{p} + \mu) \atop (b_{q} + \mu)} \right),
\end{align}
see \cite{NIST,APPrudnikov-v3}.

\item[{(iii)}] Important procedure often applied in calculations is to convert the Meijer G-function $G^{m, n}_{p, q}$ into the finite sum of the generalized hypergeometric functions ${_{p}F_{q}}$. This formula can be found in, e.g.,  \cite[Eq. (8.2.2.3)]{APPrudnikov-v3} and reads
\begin{multline}\label{14/07-6}
G^{m, n}_{p, q}\left(z\Big\vert {(a_{p}) \atop (b_{q})}\right) = \sum_{j=1}^{m} \left[\prod_{i=0}^{j-1}\Gamma(b_{i} - b_{j}) \prod_{i = j+1}^{m}\Gamma(b_{i} - b_{j})\right] \\
\times \frac{[\prod_{i=0}^{n}\Gamma(1+b_{j}-a_{i})]}{[\prod_{i = n+1}^{p} \Gamma(a_{i} - b_{j})]\, [\prod_{i=m+1}^{q} \Gamma(1 + b_{j} - b_{j})]} z^{b_{j}} \\
\times {_{p+1}F_{q}}\left({1, 1 + b_{j} - (a_{p}) \atop 1 + b_{j} - (b_{q})}; (-1)^{p-m-n} z\right),
\end{multline}
where $p \leq q$, $b_{i} - b_{j} \neq 0, \pm 1, \pm 2, \ldots$, $i\neq j$, $i, j = 1, 2, \ldots, m$. The case $p \geq q$ can be obtained from Eqs. \eqref{17/06-10} and \eqref{14/07-6}. The conditions under which Eq. \eqref{14/07-6} is held may be found in \cite[Eqs. (8.2.2.3)]{APPrudnikov-v3}.
\end{itemize}
 
\subsubsection*{The family of the Mittag-Leffler functions.}\label{secML}

The three parameter generalization of the Mittag-Leffler function \cite{book1} in the literature devoted to the fractional calculus and its applications is known as the Prabhakar function.  It was introduced by T. R. Prabhakar \cite{Prabhakar}  and reads:
\begin{equation}\label{15/06-1a}
E_{\alpha, \mu}^{\nu}(z) \okr \sum_{r \geq 0} \frac{(\nu)_{r}}{r!\,\Gamma(\alpha r + \mu)} z^r, \qquad z\in\mathbb{C},\quad \RE(\alpha) > 0, 
\end{equation}
where the Pochhammer symbol $(\nu)_{r}$ is given below Eq. \eqref{14/03-A5}. For  $\mu = \nu=1$ it reduces to the (standard, called also one-parameter) Mittag-Leffler function $E_{\alpha}(z)$ being the relaxation function of the CC model, which was first considered in \cite{Mittag2, Mittag0, Mittag1}. For $\nu=1$ it becomes the two-parameters Mittag-Leffler function  $E_{\alpha, \mu}(z)$ studied in 1905 by Wiman \cite{Wiman}.  
\begin{remark}\label{r3-20/02-23}
{\rm 
From the definition \eqref{15/06-1a} it follows that $E_{\alpha, 1}^{1}(-x^{\alpha}) \equiv E_{\alpha}(-x^{\alpha}) = 1 - x^{\alpha} E_{\alpha, 1+\alpha}(-x^{\alpha}) = \delta(x) - x^{-1} E_{\alpha, 0}(-x^{\alpha})$, where introducing the $\delta$-Dirac function allows us to avoid the pole at $x=0$.}
\end{remark}
The comprehensive review of the family of the Mittag-Leffler functions can be found in \cite{RGarra18,book1,HJHaubold11}. Here, we shall quote a few formulae employed throughout the paper.
\begin{lemma}\label{lem1}
{\rm The asymptotics of $E_{\alpha, \mu}^{\nu}(-x^{\alpha})$ for small and large $x$ is
\begin{equation}\label{26/06-10}
E_{\alpha, \mu}^{\nu}(-x^{\alpha}) \propto \frac{1}{\Gamma(\mu)} - \frac{\nu\, x^{\alpha}}{\Gamma(\alpha + \mu)}  \qquad \text{for} \quad x \ll 1
\end{equation} 
\begin{equation}\label{17/06-25}
E_{\alpha, \mu}^{\nu}(-x^{\alpha}) \propto \frac{x^{-\alpha\nu}}{\Gamma(\mu - \alpha\nu)} -  \frac{\nu x^{-\alpha\nu-\alpha}}{\Gamma(\mu - \alpha - \alpha\nu)} \qquad \text{for} \quad x \gg 1.
\end{equation}}
\end{lemma}
\begin{proof}
Asymptotics for $x \ll 1$ comes from the series expansion of three-parameters Mittag-Leffler function given by Eq. \eqref{15/06-1a}. The justification of Eq. \eqref{17/06-25} is shown in \cite{RGarra18}.
\end{proof}

For $x\in\mathbb{R}$ and $\alpha = l/k$ such that $0 < l < k$, the three-parameters Mittag-Leffler function $E_{l/k, \mu}^{\nu}(x)$ can be represented as a finite sum of $k$ generalized hypergeometric functions \cite{KGorska20a}:  
\begin{multline}\label{15/06-1}
E_{l/k, \mu}^{\nu}(x) = \sum_{j=0}^{k-1} \frac{(\nu)_{j}\, x^{j}}{j!\, \Gamma(\mu + \frac{l}{k}j)}\\ \times {_{1+k}F_{l+k}}\left({1, \Delta(k, \nu + j) \atop \Delta(k, 1+j), \Delta(l, \mu + \frac{l}{k}j)}; \frac{x^{k}}{l^{l}}\right), 
\end{multline}
where $\Delta(n, a)$ is a special list of parameters defined below Eq. \eqref{2/07-11}. 

The Meijer G-representation of  $E_{\alpha, \mu}^{\nu}(x)$ for $\alpha = l/k$ can be obtained by employing the Laplace integral of $E_{\alpha, \mu}^{\nu}(-a t^{\alpha})$. The latter is given by \cite[Eqs. (4), (5), and (8)]{KGorska21b} and/or \cite[Eqs. (A4) and (A5)]{KGorska21c}
\begin{equation}\label{17/06-1a}
E_{\alpha, \mu}^{\nu}(x) = \alpha^{-1} \int_{0}^{\infty} \E^{x \xi} \xi^{-1 -1/\alpha} \varPhi_{\alpha, \mu}^{\nu}(\xi^{-1/\alpha}) \D\xi,
\end{equation}
in which
\begin{equation}\label{17/06-2}
\varPhi_{l/k, \mu}^{\nu}(y) = \frac{k^{\nu -1}}{l^{\mu -1} \Gamma(\nu) } \frac{\sqrt{l k}}{(2\pi)^{(k-l)/2}} \frac{1}{y} G^{k, 0}_{l, k}\left(\frac{l^{l}}{k^{k} y^{l}}\Big\vert {\Delta(l, \mu) \atop \Delta(k, \nu)} \right). 
\end{equation}
We can distinguish at least two special cases of Eq. \eqref{17/06-1a} involving the one-sided L\'{e}vy stable distribution $\varPhi_{\alpha}(y) \equiv \varPhi_{\alpha, 1}^{1}(y)$ with $\alpha\in(0, 1)$ and $y\in\mathbb{R}_{+}$ related by the Laplace transform to the stretched exponential function (the KWW relaxation) \cite{KAPenson10, HPollard46}, see two examples below.
\begin{example}
{\rm The first of them is obtained from Eq. \eqref{17/06-2} by taking $\mu = \nu = 1$. Due to the definition of the Meijer G-function in the numerator of \eqref{2/08-1} the only non-vanishing terms are equal to $\prod_{i=1}^{m}\Gamma(b_{i} + s)$ whereas in the denominator survive only the terms $\prod_{i=n+1}^{p}\Gamma(a_{i} + s)$ which read $\Delta(l, 1) = \frac{1}{l}, \frac{2}{l}, \ldots, \frac{l-1}{l}, 1$ and $\Delta(k, 1) = \frac{1}{k}, \frac{2}{k}, \ldots, \frac{k-1}{k}, 1$, respectively. Delating ``$1$'' in both these sequences and adding ``$0$'' allows us to shift the upper and lower lists to $\Delta(l, 0) = 0, \frac{1}{l}, \ldots, \frac{l-1}{l}$ and $\Delta(k, 0) = 0, \frac{1}{k}, \ldots, \frac{k-1}{k}$. In consequence, we get the integral representation of the Mittag-Leffler function first obtained by H. Pollard in Ref. \cite{HPollard48} and independently derived in Refs. \cite{KGorska12, KWeron96}}:
\begin{equation}\nonumber
E_{\alpha, 1}^{1}(-x) \equiv E_{\alpha}(-x) = \int_{0}^{\infty} \E^{-x \xi} \xi^{-1-1/\alpha} \varPhi_{\alpha}(\xi^{-1/\alpha}) \frac{\D\xi}{\alpha}.
\end{equation}
\end{example}
\begin{example}
{\rm As the second example, we take $\mu = \alpha \beta$ and $\nu = \beta$.  Then $\varPhi_{\alpha, \alpha\beta}^{\beta}(y) =  \alpha y^{-\alpha\beta}\varPhi_{\alpha}(y)/\Gamma(\beta)$ for which we have \cite[Eq. (11)]{KGorska18}, this is}
\begin{equation}\nonumber
E_{\alpha, \alpha\beta}^{\beta}(-x) = \int_{0}^{\infty} \E^{-x \xi} \frac{\xi^{\beta-(1+\frac{1}{\alpha})}}{\Gamma(\beta)} \varPhi_{\alpha}(\xi^{-\frac{1}{\alpha}}) \D\xi.
\end{equation}
\end{example}

The substitution of Eq. \eqref{17/06-2} into Eq. \eqref{17/06-1a} enables us to calculate the integral in Eq. \eqref{2/07-11} and express it in the form
\begin{align}\label{17/06-3}
\begin{split}
E_{l/k, \mu}^{\nu}(-x) & = \frac{k^{\nu}}{\Gamma(\nu)} \frac{l^{\frac{l}{k}\nu - \mu + \frac{1}{2}}}{(2\pi)^{k - (1+l)/2}}\, \frac{1}{x^{\nu}}\\ & \times  G^{k, k}_{k+l, k} \left(\frac{l^{l}}{x^{k}}\Big\vert {\Delta(k, 1-\nu), \Delta(l, \mu - \frac{l}{k}\nu) \atop \Delta(k, 0)}\right) \\
& = \frac{k^{\nu}}{\Gamma(\nu)} \frac{l^{\frac{1}{2} - \mu}}{(2\pi)^{k - (1+l)/2}}\\ &\times  G^{k, k}_{k+l, k} \left(\frac{l^{l}}{x^{k}}\Big\vert {\Delta(k, 1), \Delta(l, \mu) \atop \Delta(k, \nu)}\right).
\end{split}
\end{align}
Employing Eqs. \eqref{17/06-10} and \eqref{14/07-6} we can express Eq. \eqref{17/06-3} in terms of the generalized hypergeometric functions, {\textit {vis.}} Eq. \eqref{15/06-1}. 

The Laplace transform of the three-parameters Mittag-Leffler function \cite[Eq. (2.5)]{Prabhakar}:
\begin{equation}\label{15/06-4}
\mathscr{L}[t^{\mu - 1} E_{\alpha, \mu}^{\nu}(- \lambda t^{\alpha}); z] = \frac{z^{\alpha\nu - \mu}}{(\lambda + z^{\alpha})^{\nu}}, \qquad \lambda \in \mathbb{R},
\end{equation}
is crucially important for solving problems which appear in the relaxation theory, e.g., it allows us to calculate the response function $\phi(t)$ from the experimental data encoded in phenomenological spectral functions $\widehat{\phi}(s)$ whose forms, see  Eq. \eqref{23/05_2}, match the RHS of Eq. \eqref{15/06-4}. Both heuristic, and rigorous arguments quoted in Sect. \ref{sec2} suggest that looking for the relaxation functions we should focus our attention on the functions belonging to the CM class.
Here it has to be said that physicists' interest in Mittag-Leffler functions as elements of the CM class is by no means limited to dielectric relaxation phenomena \cite{ECapelas11,RGarrappa16, Hanyga1,FMainardi15,Tomovski14} but is also well-grounded in studies of other phenomena exhibiting  the memory-dependent time evolution, to mention hereditary mechanics and viscoelasticity \cite{Hanyga2, Mainardibook,Rabotnov,Rossikhin3, Rossikhin2}. Mathematicians' interest to study CM property of the Mittag-Leffler functions stems from the classical works of H. Pollard and W. R. Schneider who proved it for $E_{\alpha}(-x)$ \cite{HPollard48} and $E_{\alpha,\beta}(-x)$ \cite{WRSchneider96}. Currently the research in the field is inspired by problems and applications of the special functions theory \cite{Samko,Shukla,Simon}, fractional calculus and fractional differential equations \cite{book1,Mainardibook}. Recent results include  proofs of the CM character obeyed by the Prabhakar function \cite{FMainardi15}, see the lemma below, and by $E_{\alpha, \mu}^{\nu}(-a t^{\alpha})$ itself if $0 < \alpha \leq 1$, $\mu \geq \alpha\nu$, and $\gamma >0$, as shown in \cite{KGorska21b}. CM property of  $E_{\alpha, \mu}^{\nu}(-a t^{\alpha})$ explains restriction $\beta<1$ put on the parameter $\beta$ in the Prabhakar function to preserve its CM as a result of the multiplication of two CM functions.      

\begin{lemma}\label{lem3}
{\rm $t^{\mu - 1} E_{\alpha, \mu}^{\nu}(-a t^{\alpha})$ is CM function for $0 < \alpha \leq 1$, $0 < \mu \leq \alpha\nu < 1$, and $\nu > 0$.}
\end{lemma}
\begin{proof}
Different proofs of this lemma were presented in \cite{KGorska21b, FMainardi15,Tomovski14}.
\end{proof}
\begin{lemma}\label{lem2}
{\rm $x^{-\alpha\beta-1} E^{\beta}_{-\alpha, -\alpha\beta}(-\lambda x^{-\alpha})$ can be replaced by $\lambda^{-\beta} x^{-1} E_{\alpha, 0}^{\beta}(-x^{\alpha}/\lambda)$, $\lambda \in \mathbb{R}$. }
\end{lemma}
\noindent
{\em Proof.}  Lemma \ref{lem2} comes from direct calculations:
\begin{align*}
\mathscr{L}[x^{-\alpha\beta-1} & E_{-\alpha, -\alpha\beta}^{\beta}(-\lambda x^{-\alpha}); z]  \\ & = \frac{1}{(z^{-\alpha} + \lambda)^{\beta}} = \lambda^{-\beta} \frac{z^{\alpha\beta}}{(\lambda^{-1} + z^{\alpha})^{\beta}} \\ & = \mathscr{L}^{-1}[\lambda^{-\beta} x^{-1} E_{\alpha, 0}^{\beta}(-x^{\alpha}/\lambda); z]. \qquad \qed
\end{align*}
\begin{lemma}\label{lem3-1}
{\rm $$E^{\beta}_{\alpha, 1}(-\lambda x^{\alpha}) = \lambda^{-\beta} \int_{0}^{x} u^{-\alpha\beta - 1}  E_{-\alpha, -\alpha\beta}^{\beta}(-u^{-\alpha}/\lambda) \D u,$$ $\lambda \in \mathbb{R}$. }
\end{lemma}
\noindent
{\em Proof.} Lemma \ref{lem3-1} can be shown by making the direct calculations analogical like in the proof of Lemma \ref{lem2} and using $\mathscr{L}^{-1}[z^{-1} f(z); x] = \int_{0}^{x} \mathscr{L}^{-1}[f(z); u] \D u$. \qed

We also mention that the fractional derivative in Riemann-Liouville sense of $x^{\mu-1}E_{\alpha, \mu}^{\nu}(\lambda x^{\alpha})$ yields \cite[Eq. (5.1.34)]{book1}, namely
\begin{equation}\label{14/07-2}
\left(D^{\beta}[y^{\mu-1}E_{\alpha, \mu}^{\nu}(\lambda y^{\alpha})]\right)\!(x) = x^{\mu-\beta-1} E_{\alpha, \mu-\beta}^{\nu}(\lambda x^{\alpha}).
\end{equation}

\section{The Havriliak-Negami model and its physical content}\label{sec5}

In this section, we shall illustrate the general theoretical approach presented in Secs. \ref{SEC1}-\ref{SEC3} on the example of the HN relaxation model and its special cases, like the Debye, CC and CD models. The HN relaxation pattern was introduced in \cite{SHavriliak67} to parametrize experimental data describing the frequency dependence  of the complex dielectric permittivity $\widehat{\varepsilon^{\star}}(\omega)$ measured in polymers. Despite its purely phenomenological origin and apparent simplicity, the applicability of the HN model went far beyond its initial implementation. The model has appeared well-working for a much larger plethora of dielectric phenomena and has become  the ``first choice'' method to analyse experimental relaxation data for various relaxation phenomena taking place in different complex systems, by no mere limited to those of condensed matter physics and materials science. Unexpected examples of its utility include the use of the HN function for monitoring the contamination in sandstone \cite{VSaltas07}, or investigations of complex systems representing plant tissues of fresh fruits and vegetables, for which the HN relaxation in the frequency range $10^{7}-1.8 \times10^{9}$ Hz was shown to be an useful tool of analysis \cite{RRNigmatulin06}. 
 
\subsection{The spectral function $\widehat{\phi}_{H\!N}(\alpha, \beta; \I\!\omega)$}\label{sec5.1}

We start our consideration by recalling the spectral function for the HN model 
\begin{equation}\label{3/06-2}
\widehat{\phi}_{H\!N}(\alpha, \beta; \I\!\omega) = \left[\frac{\widehat{\varepsilon^{\star}}(\omega) - \varepsilon_{\infty}}{\varepsilon - \varepsilon_{\infty}}\right]_{H\!N} =  [1 + (\I\!\omega\tau)^{\alpha}]^{-\beta},
\end{equation}
$\alpha, \beta\in(0, 1]$, given by Eqs. \eqref{23/05_2}. The range of $\alpha$ and $\beta$ parameters (called respectively the width and asymmetry) are indicated by the experiment which implies that they belong to $(0, 1]$. We recall that for $\alpha = \beta = 1$ we have the Debye spectral function, for $\alpha\in(0, 1)$ and $\beta = 1$ it reduces to the CC case, and for $\alpha = 1$ and $\beta\in(0, 1)$ becomes the CD model. 

Separating real and imaginary parts of $\widehat{\varepsilon^{\star}}_{H\!N}(\omega) = \widehat{\varepsilon'}_{H\!N}(\omega) - \I \widehat{\varepsilon''}_{H\!N}(\omega)$ we get after some complex algebra (see Ref. \cite{CJFBottcher78}): 
\begin{equation}\nonumber
\widehat{\varepsilon'}_{H\!N}(\omega) = \varepsilon_{\infty} + \frac{(\varepsilon - \varepsilon_{\infty})\cos\big[\beta\,\theta_{\alpha/2}\big((\omega t)^{2}\big)\big]}{\big[1 + 2(\omega\tau)^{\alpha}\cos(\pi\alpha/2) + (\omega\tau)^{2\alpha}\big]^{\beta/2}},
\end{equation}
and
\begin{equation}\nonumber
\widehat{\varepsilon''}_{H\!N}(\omega) = \frac{(\varepsilon - \varepsilon_{\infty})\sin\big[\beta\,\theta_{\alpha/2}\big((\omega t)^{2}\big)\big]}{\big[1 + 2(\omega\tau)^{\alpha}\cos(\pi\alpha/2) + (\omega\tau)^{2\alpha}\big]^{\beta/2}},
\end{equation}
where
\begin{equation}\label{3/06-4}
\theta_{\alpha}(y) = \arctan\left(\frac{\sin(\pi\alpha)}{y^{-\alpha} + \cos(\pi\alpha)}\right)
\end{equation}
provides restrictions on the parameters of the model \cite{FMainardi15}. 

\subsection{The response function $\phi_{H\!N}(\alpha, \beta; t)$}\label{sec5.2}

To find the response function $\phi_{H\!N}(\alpha, \beta; t)$ firstly we set $z = \I\!\omega$ in Eq. \eqref{3/06-2} and calculate the inverse Laplace transform of $\widehat{\phi}_{H\!N}(\alpha, \beta; z)$. Then, with the help of Eq. \eqref{15/06-4} we get
\begin{align}\label{16/06-21}
\begin{split}
\phi_{H\!N}(\alpha, \beta; t) &= \tau^{-\alpha\beta} \mathscr{L}^{-1}[(\tau^{-\alpha} + s^{\alpha})^{-\beta}; t]   \\
&= \frac{1}{\tau}\, \big(\ulamek{t}{\tau}\big)^{\alpha\beta - 1} E^{\beta}_{\alpha, \alpha\beta}\big(\!-\big(\ulamek{t}{\tau}\big)^{\alpha}\big),
\end{split}
\end{align}
where $E_{\alpha, \alpha\beta}^{\beta}(-(t/\tau)^{\alpha})$ is the three-parameter Mittag-Leffler function described in Sect. \ref{secML}. For $\alpha = \beta = 1$ it reduces to the Debye response function $\phi_{H\!N}(1, 1; t) \equiv \phi_{D}(t) = \exp(-t/\tau)/\tau$ while for $\alpha = 1$ and $\beta\in(0, 1)$ we get the CD response function  
\begin{align}\nonumber
\begin{split}
\phi_{H\!N}(1, \beta; t) \equiv \phi_{C\!D}(\beta; t) & = \tau^{-1} \big(\ulamek{t}{\tau}\big)^{\beta - 1} E_{1, \beta}^{\beta}\big(\!-\!\big(\ulamek{t}{\tau})\big) \\ &= \frac{(t/\tau)^{\beta - 1}}{\tau \Gamma(\beta)} \E^{-t/\tau},
\end{split}
\end{align}
and for $\alpha\in(0, 1)$ and $\beta = 1$ the CC response function 
\begin{equation}\nonumber
\phi_{H\!N}(\alpha, 1; t) \equiv \phi_{CC}(\alpha; t) = \tau^{-1}\big(\ulamek{t}{\tau}\big)^{\alpha - 1} E_{\alpha, \alpha}\big(\!-\!\big(\ulamek{t}{\tau}\big)^{\alpha}\big).
\end{equation}
For rational $\alpha = l/k$ the above formulae can be expressed in terms of special functions directly implemented in the CAS, like the Meijer G-function and/or the finite sum of the generalized hypergeometric functions ${_{p}F_{q}}$:
\begin{align}\label{17/06-21}
\phi_{H\!N}&(l/k, \beta; t) = (2\pi)^{\ulamek{1+l}{2} - k} \frac{\sqrt{l} k^{\beta}}{\Gamma(\beta)}\, \frac{1}{t}\\ & \times G^{k, k}_{l+k, k}\left(\frac{l^{l}\tau^{l}}{t^{l}}\Big\vert{\Delta(k, 1-\beta), \Delta(l, 0) \atop \Delta(k, 0)}\right)  \nonumber \\ 
& = \frac{1}{\tau} \sum_{j=0}^{k-1} \frac{(-1)^{j} (\beta)_{j}}{j!\, \Gamma(\frac{l}{k}(\beta+j))} \left(\frac{t}{\tau}\right)^{\!\!\frac{l}{k}(\beta + j) - 1}\label{17/06-20}  \\ \nonumber & \times {_{1+l}F_{l + k}}\left({1, \Delta(k, \beta + j) \atop \Delta(k, 1+j), \Delta(l, \frac{l}{k}(\beta+j))}; \frac{(-1)^{k} t^{l}}{l^{l} \tau^{l}}\right).
\end{align}
In Fig. \ref{rys1} we present $\phi_{H\!N}(\alpha, \beta; t)$ for $\alpha= 1/2$ and $\beta = 1/2, 1, 2$ as well as for $\beta=3$. 
\begin{figure}[!h]
\begin{center}
\includegraphics[scale=0.42]{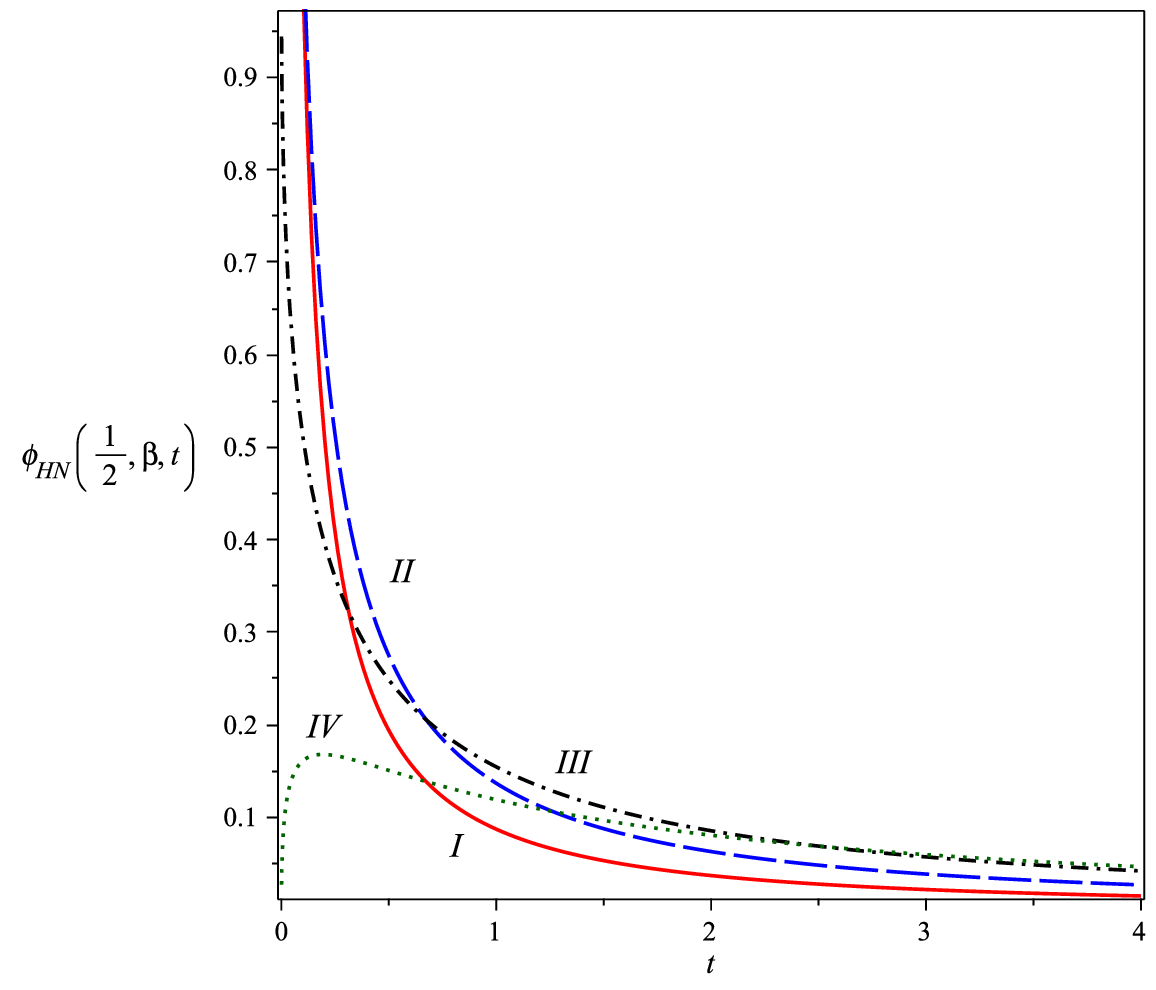}
\caption{\label{rys1}(Color online) Plot of $\phi_{H\!N}(1/2, \beta; t)$ given by Eq. \eqref{17/06-21} for $\tau = 1$ and $\beta = 1/2$ (the red continuous curve no. I), $\beta = 1$ (the blue dashed curve no. II), $\beta = 2$ (the black dash-dotted curve no. III), and $\beta = 3$ (the green dotted curve no. IV). } 
\end{center}
\end{figure}
It is seen that for $\beta = 3$, $\phi_{H\!N}(1/2, 3; t)$ is unimodal with the maximum at $t = t_{\rm max}$. Thus it can not be CM function. Looking at plots in Fig. \ref{rys1}, one may see the candidates to be a CM function are $\phi_{H\!N}(1/2, \beta; t)$ with $\beta \leq 2 = 1/\alpha$. This is confirmed by Lemma \ref{lem2} applied for Eq. \eqref{16/06-21} but the case III has to be excluded as experimental data restrict values of $\beta$ to the interval $(0, 1]$. 

For reader's convenience we quote the asymptotic behavior of response function $\phi_{H\!N}(\alpha, \beta; t)$ for a small and large time presented also in \cite[Eq. (3.29)]{RGarrappa16}: 
\begin{multline}\nonumber
\phi_{H\!N}(\alpha, \beta; t) \propto \frac{1}{\tau \Gamma(\alpha\beta)} \left(\frac{t}{\tau}\right)^{\alpha\beta-1} \quad \text{for} \quad \frac{t}{\tau} \ll 1 \quad \text{and} \\
\phi_{H\!N}(\alpha, \beta; t) \propto - \frac{\beta}{\tau \Gamma(-\alpha)} \left(\frac{t}{\tau}\right)^{-1-\alpha}\quad \text{for} \quad \frac{t}{\tau} \gg 1.
\end{multline}
The asymptotics at $t \to 0$ depends on the value of $\beta$ namely $\phi_{H\!N}(\alpha, \beta; t)$ tends to infinity for $\beta < 1/\alpha$ (note that in such a case the singularity is integrable), goes to $\tau^{-1}$ for $\beta = 1/\alpha$, and approaches 0 for $\beta > 1/\alpha$.
 
\subsection{The probability density $g_{H\!N}(\alpha, \beta; \xi)$}\label{sec5.3}

Thanks to Theorem 1 (Sect. \ref{subsec3.1}) it turns out that the spectral function has the integral representation as in Eq. \eqref{6/05}: 
\begin{equation}\label{18/06-5}
\widehat{\phi}_{H\!N}(\alpha, \beta; z) = \int_{0}^{\infty} \frac{p_{H\!N}(\alpha, \beta; \xi)}{\xi + z} \D\xi,
\end{equation}
where $p_{H\!N}(\alpha, \beta; \xi)$ is PDF. Restricting the complex $z$ onto $s\in\mathbb{R}_{+}$ we observe that $\widehat{\phi}_{H\!N}(\alpha, \beta; z)$ is a S function on the real axis and, thus, we confirm the information presented in Ref. \cite{KGorska21c}. From the Schwinger parametrization it follows that Eq. \eqref{18/06-5} is the Laplace transform of the Laplace transform, namely $\widehat{\phi}_{H\!N}(\alpha, \beta; z) = \mathscr{L}[\phi_{H\!N}(\alpha, \beta; t); z]$ where $\phi_{H\!N}(\alpha, \beta; t) = \mathscr{L}[p_{H\!N}(\alpha, \beta; \xi); t]$. Because $\phi_{H\!N}(\alpha, \beta; t)$ is a CM function for $\alpha \in (0, 1]$ and $\beta\le 1/\alpha$, then from the Bernstein theorem it emerges that $p_{H\!N}(\alpha, \beta; \xi)$ is a non-negative function for 
\begin{equation}\label{18/06-20}
0 < \alpha < 1 \quad \text{and} \quad 0 < \beta \leq 1/\alpha.
\end{equation}

Furthermore, from Eq. \eqref{18/06-11} and $\phi_{H\!N}(\alpha, \beta; t) = {\mathcal{L}} [p_{H\!N}(\alpha, \beta; \xi);t]$ we get 
\begin{align}\label{18/06-12}
\begin{split}
g_{H\!N}(\alpha, \beta; \xi) & = \frac{\tau}{\xi}\, p_{H\!N}(\alpha, \beta; \xi) \\
& = \frac{\tau}{\xi}\, \mathscr{L}^{-1}[\phi_{H\!N}(\alpha, \beta; t); \xi],
\end{split}
\end{align}
being also a non-negative function, since the product of the non-negative functions is also non-negative. That statement agrees with one of the results presented in \cite[Sect. 4]{KGorska18}. 

For the rational $\alpha$ such that $0 < \alpha = l/k < 1$, we substitute $\phi_{H\!N}(\alpha, \beta; t)$ in Eq. \eqref{17/06-21} and employ Eq. \eqref{2/07-11}. Hence, we conclude that $g_{H\!N}(l/k, \beta; \xi)$ can be written as 
\begin{align}\label{5/06-1}
\begin{split}
g_{H\!N}(&l/k, \beta; \xi) = (2\pi)^{\frac{1+l}{2}-k} \frac{\sqrt{l} k^{\beta}}{\Gamma(\beta)}\,\frac{1}{\xi}\\
& \times \mathscr{L}^{-1}\left[\frac{\tau}{t} G^{k, k}_{l+k, k}\left(\frac{l^{l} \tau^{l}}{t^{l}}\Big\vert{\Delta(k, 1-\beta), \Delta(l, 0) \atop \Delta(k, 0)}\right); \xi\right]\\
&= (2 \pi)^{l-k}\frac{k^{\beta}}{\Gamma(\beta)}\, \frac{1}{\xi}\, G^{k, k}_{l+k, l+k}\left(\xi^{l}\Big\vert{\Delta(k, 1-\beta), \Delta(l, 0) \atop \Delta(k, 0), \Delta(l, 0)}\right),
\end{split}
\end{align}
with the symbols $\Delta(n, a)$ defined below Eq. \eqref{15/06-1}. Next, using Eqs. \eqref{17/06-10} and \eqref{14/07-6} as well as the Gauss-Legendre multiplication formula for $\Gamma$ functions in Eq. \eqref{5/06-1}, we can express $g_{H\!N}(l/k, \beta; \xi)$ as a finite sum of $k$ generalized hypergeometric functions
\begin{align}\label{5/06-2}
\begin{split}
g_{H\!N}(l/k, \beta; \xi) &= \frac{1}{\pi} \sum_{j=0}^{k-1} \frac{(-1)^{j}}{j!} \frac{(\beta)_{j}}{\xi^{1+\frac{l}{k} n_j}} \sin(\ulamek{l}{k}n_j\pi)\\
& \times {_{k+1}F_{k}}\left({1, \Delta(k, n_{j}) \atop \Delta(k, 1+j)}; \frac{(-1)^{l-k}}{\xi^{l}}\right)
\end{split}
\end{align} 
with $n_j = \beta + j$. Equation \eqref{5/06-2} gives the form of $g_{H\!N}(l/k, \beta; \xi)$ which is convenient and efficiently applicable in calculations using the standard computer algebra packages. For example, for the CC relaxation ($\beta = 1$) it is seen that, due to appropriate cancellations,  ${_{k+1}F_{k}}$'s reduce to ${_{1}F_{0}}\left({1 \atop 0}, u\right)$, which can be further simplified to $(1 - u)^{-1}$ where $u = (-1)^{l-k}\xi^{-l/k}$. Employing it and Eqs. (1.353.1) and (1.353.3) on p. 38 of \cite{Gradshteyn07} to the sum over $j$, we get 
\begin{align}\label{16/06-1}
\begin{split}
g_{H\!N}(\alpha, 1; \xi) & \equiv g_{CC}(\alpha; \xi) \\ &= \frac{\xi^{\alpha -1} \sin(\alpha\pi)}{\pi (\xi^{2\alpha} + 2 \xi^{\alpha} \cos(\alpha\pi) + 1)},
\end{split}
\end{align}
with $0 <\alpha = l/k < 1$. We point out that Eq. \eqref{16/06-1} was obtained in \cite{ECapelasDeOliveira14,BDybiec10,RGorenflo08,KWeron96} using different methods, see Eq. (3.24) on p. 245 in \cite{RGorenflo08}, Eqs. (22) and (39) in \cite{ECapelasDeOliveira14} or Eq. (26) in \cite{BDybiec10}. Distributions  $g_{CC}(\alpha; \xi)$ are non-negative functions for $\alpha\in(0, 1)$ and they share the following properties: (i) $g_{CC}(\alpha; \xi)$ are non-negative and non-increasing for $0 < \alpha < 1$ and $\xi \geq 0$; and (ii) $g_{CC}(\alpha; \xi)$ go to infinity at $\xi = 0$, and vanish for $\xi\to\infty$. For the CD relaxation ($\alpha = 1$) in Eq. \eqref{5/06-2} we take $l = k = 1$ which leads to \cite[Eq. (3.19)]{RGarrappa16}:
\begin{equation}\nonumber
g_{H\!N}(1, \beta; \xi) \equiv g_{C\!D}(\beta; \xi) = \frac{\sin(\beta\pi)}{\pi\xi (\xi-1)^{\beta}}\, \Theta(\xi-1),
\end{equation}
with the Heaviside step function $\Theta(\cdot)$ which guarantees that $g_{C\!D}(\beta; \xi)$ is real for $\xi \leq 1$.

Our last task in this subsection is to find the series representation of $g_{H\!N}(\alpha, \beta; \xi)$ for $\alpha = l/k$. To achieve this goal, we use the series representation of the generalized hypergeometric functions ${_{p}F_{q}}$, see Eq. \eqref{14/03-A5}. Following such a way we get Eq. \eqref{5/06-2} as the double sum: one over $r$ ($r=0, 1, 2, \ldots$) which comes from the series representation of ${_{p}F_{q}}$, and the second one over $j$ ($j = 0, 1, \ldots, k-1$) which appears in the Eq. \eqref{5/06-2} itself. Changing the summation index $k r + j \to n$ we arrive at the expression 
\begin{equation}\label{6/06-4}
g_{H\!N}(\alpha, \beta; \xi) = \frac{1}{\pi}\, \sum_{n=0}^{\infty} \frac{(-1)^{n}}{n!} \frac{ (\beta)_{n} }{\xi^{1+\alpha(\beta+n)}} \sin[\alpha(\beta+n) \pi],
\end{equation}
which, after representing the sine function as an imaginary part of $\exp[\I\!\pi \alpha (\beta+n)]$, using the integral representation of the $\Gamma$ function, and applying Eq. \eqref{15/06-1a}, leads us to the integral form of Eq. \eqref{6/06-4}:
\begin{align}\label{6/06-5}
\begin{split}
g_{H\!N}(\alpha, \beta; \xi) &= \frac{1}{\pi} \IM\left\{\E^{\I\!\alpha\beta\pi} \int_{0}^{\infty} e^{-u\xi} u^{\alpha\beta} \right. \\ & \times \left.E_{\alpha, \alpha\beta}^{\beta}(- u^{\alpha} \E^{\I\!\alpha\pi}) \D u\right\}. 
\end{split}
\end{align}

Applying Eq. \eqref{15/06-4} to Eq. \eqref{6/06-5} and employing de Moivre's formula to calculate the imaginary part, we rederive the function $g_{H\!N}(\alpha, \beta; \xi)$ in the form obtained and extensively studied in the just quoted Ref. \cite{FMainardi15}. Namely, for $0 < \alpha < 1$ and $0 < \beta \leq 1/\alpha$, we get two solutions
\begin{equation}\label{14/06-2}
g_{H\!N}(\alpha, \beta; \xi) = \pm\frac{1}{\pi \xi} \frac{\sin\left[\beta\,\theta_{\alpha}(\xi) \right]}{[\xi^{2\alpha} + 2\xi^{\alpha}\cos(\pi\alpha) + 1]^{\beta/2}}
\end{equation}
where $\theta_{\alpha}(y)$ is defined in Eq. \eqref{3/06-4} and the sign in Eq. \eqref{14/06-2} depends on the choice of the branch of the arctan function in  Eq. \eqref{3/06-4} having the essential singularity for $\xi = [\cos(\pi\alpha+\pi)]^{-1/\alpha}$. Equation \eqref{14/06-2} for $\beta = 1$ is identically equal to Eq. \eqref{16/06-1}. In addition, Eq. \eqref{18/06-20} implies the non-negativity of Eq. \eqref{14/06-2}. The denominator in Eq. \eqref{14/06-2} is always positive so $g_{H\!N}(\alpha, \beta; \xi)$ remains non-negative if $\beta\theta_{\alpha}(1)\in[0, \pi]$ mod $2n\pi$. Since Eq. \eqref{3/06-4} gives $\theta_{\alpha}(1)\in [0, \pi\alpha]$, this leads to $\beta \leq 1/\alpha$. This fully agrees with the observations made in {\cite[Eq. (3.47)]{RGarrappa16}} and \cite[Eq. (8)]{FMainardi15}. However, $0 < \alpha, \beta < 1$ so that  $\beta \leq 1/\alpha$ is always satisfied for these range of parameters $\alpha$ and $\beta$.

Fig. \ref{rys5}a illustrates $g_{H\!N}(\alpha, \beta; \xi)$ as the function of $\xi$ for fixed $\alpha$ and different values of $\beta$. The opposite situation is presented in Fig. \ref{rys5}b where $\beta$ is fixed while $\alpha$ is changing. To get both plots we can use either Eq. \eqref{14/06-2} which contains the more friendly elementary functions, or the Meijer G-representation of $g_{H\!N}(\alpha, \beta; \xi)$ given by Eq. \eqref{5/06-1}. In the first case we should separate the range of $\xi$ into two sectors: $\xi\in[0, \xi_{p})$ and $\xi\in[\xi_{p}, \infty)$. At the point $\xi_{p}$ $g_{H\!N}(\alpha, \beta; \xi)$ changes the sign such that we take $g_{H\!N}(\alpha, \beta; \xi)$ with a plus sign in the range $\xi\in[0, \xi_{p})$ and $g_{H\!N}(\alpha, \beta; \xi)$ with a minus sign for $\xi\in[\xi_{p}, \infty)$. In the Meijer G-case, the required separation is done automatically via the use of CAS.
\begin{figure}[!h]
\includegraphics[scale=0.4]{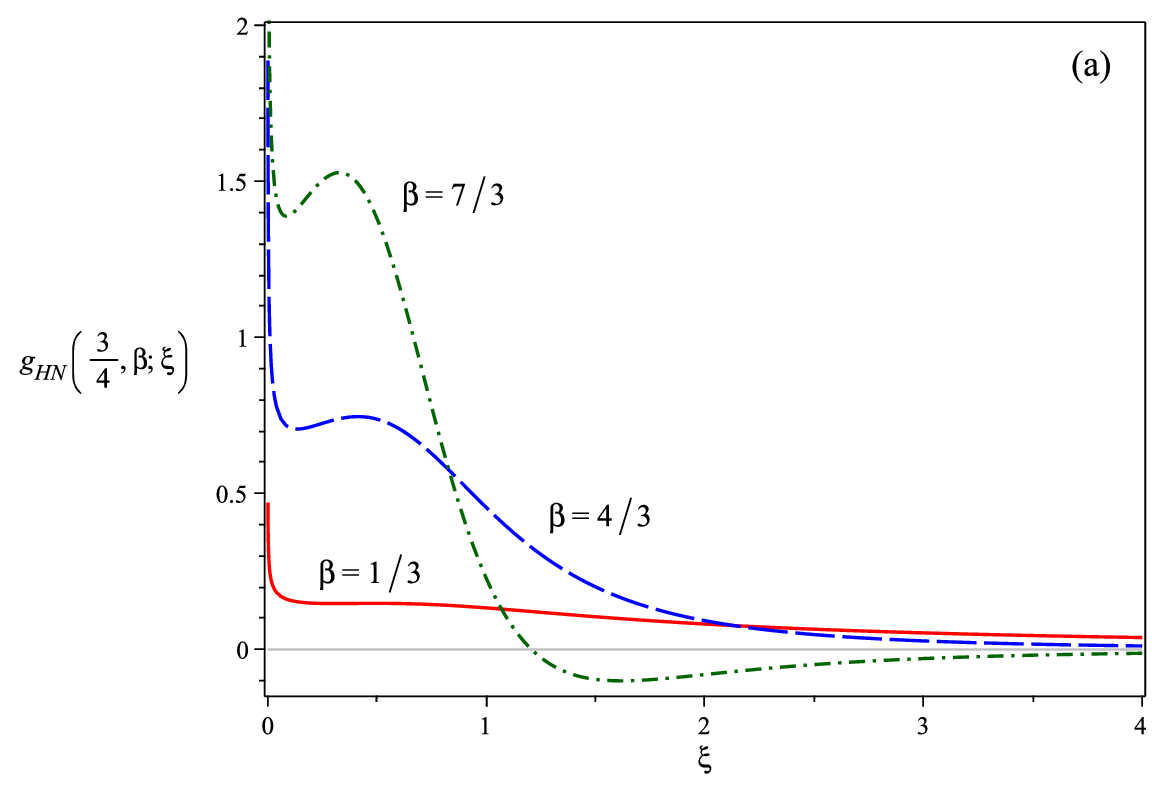}
\includegraphics[scale=0.4]{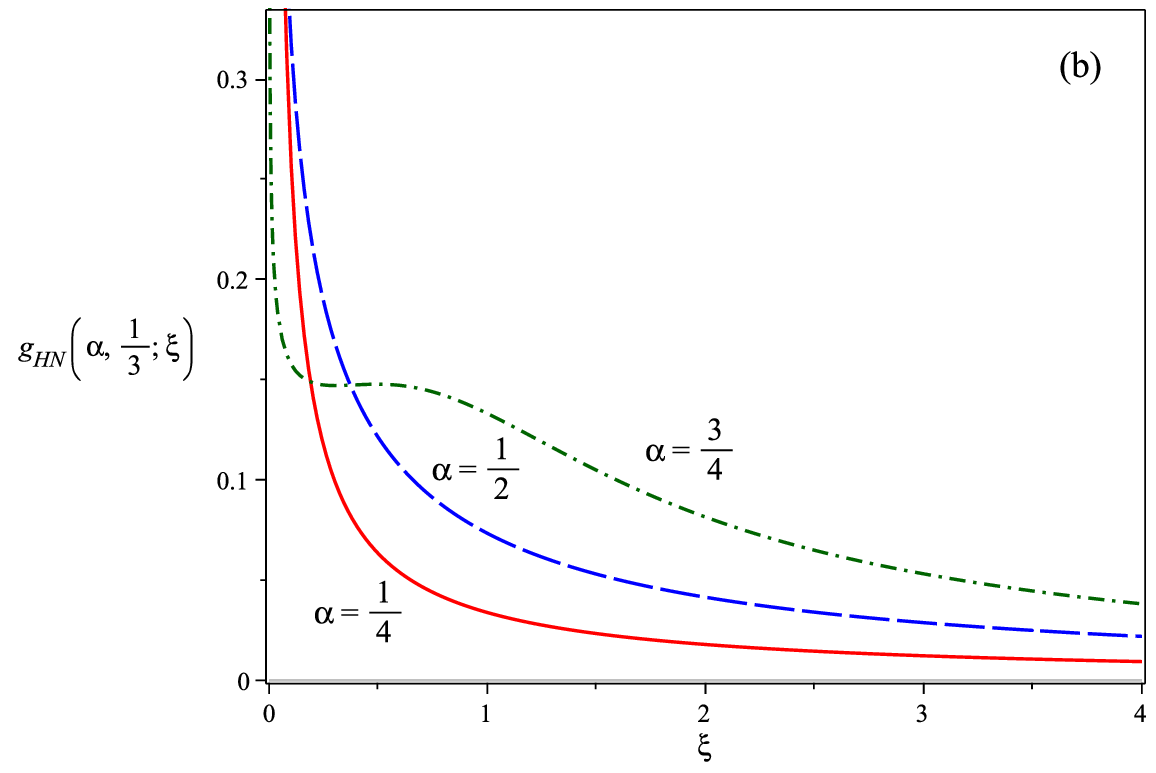}
\caption{\label{rys5}(Color online) Plot of $g_{H\!N}(\alpha, \beta; \xi)$  as a function of $\xi$ given by Eq. \eqref{5/06-1} for $\alpha, \beta = {\rm const}$. In Fig. \ref{rys5}a we take $\alpha = 3/4$ and $\beta = 1/3$ (the red continuous curve), $\beta = 4/3$ (the blue dashed curve), and $\beta = 7/3$ (the green dashed-dotted curve). It should be observed that $g_{H\!N}(\alpha, \beta; \xi)$ contains a negative part for $\beta > 1/\alpha$.  In Fig. \ref{rys5}b we have $\beta = 1/3$ and $\alpha = 1/4$ (the red continuous curve), $\alpha = 1/2$ (the blue dashed curve), and $\beta = 3/4$ (the green dashed-dotted curve).} 
\end{figure}

\subsection{The relaxation function $n_{H\!N}(\alpha, \beta; t)$}\label{sec5.4}

The HN relaxation function may be obtained in at least in two ways: (i) by inserting $g_{H\!N}(\alpha, \beta; t)$ into Eq. \eqref{18/06-10} or (ii) by calculating the inverse Laplace transform of $\widehat{n}_{H\!N}(\alpha, \beta; z)$ where $\widehat{n}_{H\!N}(\alpha, \beta; z)$ is derived from Eq. \eqref{23/05_1} adjusted for the HN model. Both these ways lead to the same results, so we choose (ii) as more convenient for us. Here, $n_{H\!N}(\alpha, \beta; t)$ reads
\begin{align}\label{19/06-1}
\begin{split}
n_{H\!N}(\alpha, \beta; t) & = \mathscr{L}^{-1}[\widehat{n}_{H\!N}(\alpha, \beta; z); t] = 
\\
& =\mathscr{L}^{-1}[z^{-1}; t] - \tau^{-\alpha\beta}\mathscr{L}^{-1}\left[\frac{z^{-1}}{(\tau^{-\alpha} + z^{\alpha})^{\beta}}; t\right] \\
& = 1 - \big(\ulamek{t}{\tau}\big)^{\alpha\beta} E_{\alpha, 1 + \alpha\beta}^{\beta}\big(\!-(\ulamek{t}{\tau})^{\alpha}\big),
\end{split}
\end{align}
where $z = \I\!\omega$ and we used the Laplace transform given by Eq. \eqref{15/06-4}. To express $n_{H\!N}(\alpha, \beta; t)$ in the language of the Meijer G-function we use Eq. \eqref{17/06-3}, which for $\alpha = l/k$ allows us to write
\begin{align}\label{27/06-1a}
\begin{split}
n_{H\!N}(\alpha, \beta; t) & = 1 - (2\pi)^{\frac{1+l}{2} - k}\,\frac{k^{\beta}l^{-1/2}}{\Gamma(\beta)} \\ & \times G_{l+k, k}^{k, k}\left(\frac{l^{l} \tau^{l}}{t^{l}}\Big\vert {\Delta(k, 1-\beta), \Delta(l, 1) \atop \Delta(k, 0)}\right).
\end{split}
\end{align}
The relaxation function $n_{H\!N}(\alpha, \beta; t)$ given by Eq. \eqref{27/06-1a} is plotted in Fig. \ref{rys2} for $\alpha = 1/2$ and $\beta = 1/4, 1/2$, and $3/4$.
\begin{figure}[!h]
\begin{center}
\includegraphics[scale=0.42]{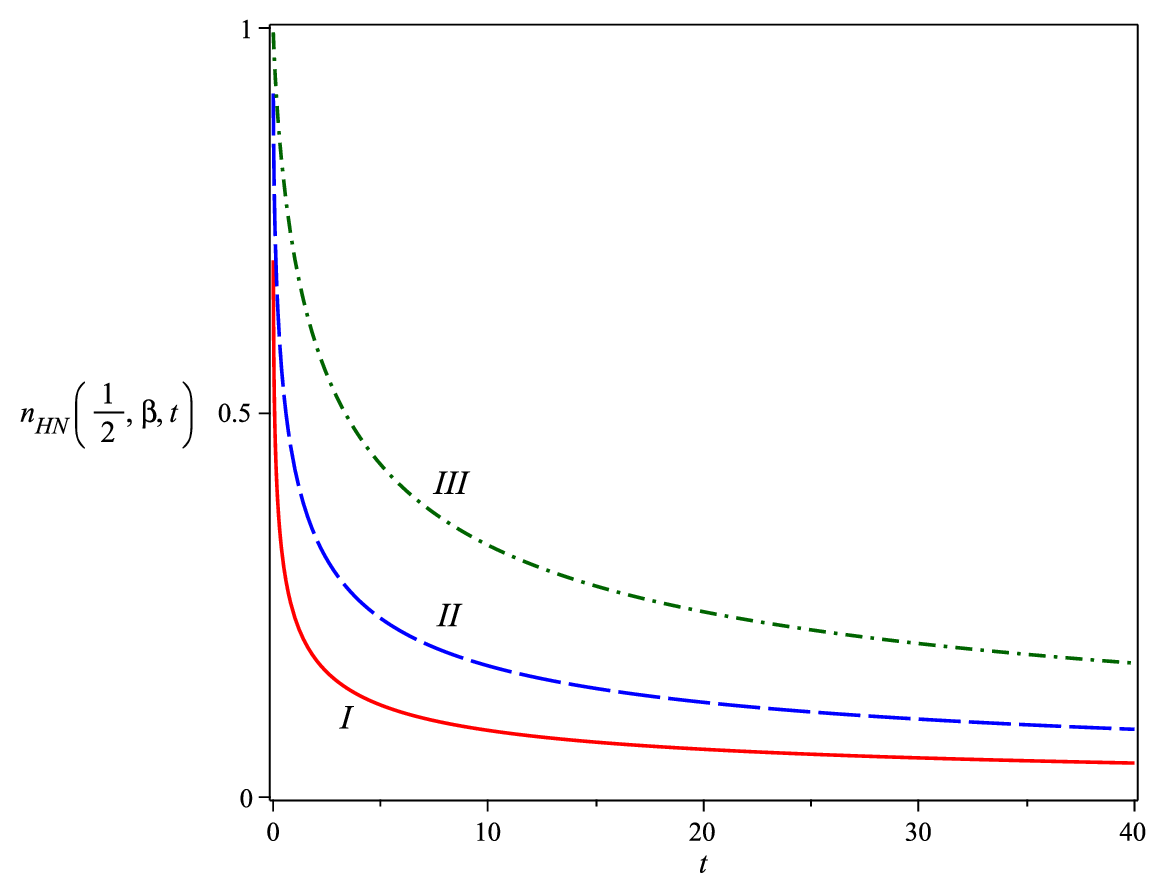}
\caption{\label{rys2}(Color online) Plot of $n_{H\!N}(1/2, \beta; t)$ as the function of $t$ for $\tau = 1$ and $\beta = 1/2$ (the red continuous curve no. I), $\beta = 1$ (the blue dashed curve no. II), and $\beta = 2$ (the green dash-dotted curve no. III). } 
\end{center}
\end{figure}
Using appropriate representation of the Mittag-Leffler function given by Eq. \eqref{17/06-3} we can express $n_{H\!N}(\alpha, \beta; t)$ in terms of these functions. So,  for $\alpha\in(0, 1)$ and $\beta = 1$, Eq. \eqref{19/06-1} leads to 
\begin{align}\label{3/08-2}
\begin{split}
n_{H\!N}(\alpha, 1; t) \equiv n_{CC}(\alpha; t) & = 1 - E^{1}_{\alpha, 1 +\alpha}\big(\!-\!(\ulamek{t}{\tau})^{\alpha}\big) \\ & = E_{\alpha}\big(\!-\!(\ulamek{t}{\tau})^{\alpha}\big).
\end{split}
\end{align} 
For $\alpha = 1$ and $\beta\in(0, 1)$ it gives 
\begin{align}\label{3/08-3}
\begin{split}
n_{H\!N}(1, \beta; t) \equiv n_{CD}(\beta; t) & = 1 - E^{\beta}_{1, 1 + \beta}\big(\!-\!\ulamek{t}{\tau}\big) \\ &= \Gamma\big(\beta,\ulamek{t}{\tau}\big)/\Gamma(\beta),
\end{split}
\end{align}
where $\Gamma(a, y) = \int_{y}^{\infty} u^{\alpha-1} \E^{-u} \D u$ is the upper incomplete gamma function \cite{NIST}.

Applying Lemma \ref{lem1} to Eq. \eqref{19/06-1} we find the asymptotics of $n_{H\!N}(\alpha, \beta; t)$ which is proportional to 
\begin{equation}\nonumber
n_{H\!N}(\alpha, \beta; t) \propto 1 - \frac{1}{\Gamma(1 + \alpha\beta)} \left(\frac{t}{\tau}\right)^{\alpha\beta} \quad \text{for} \quad \frac{t}{\tau} \ll 1 
\end{equation}
and
\begin{equation}\nonumber
n_{H\!N}(\alpha, \beta; t) \propto \frac{\beta}{\Gamma(1-\alpha)} \left(\frac{t}{\tau}\right)^{-\alpha} \quad \text{for} \quad \frac{t}{\tau} \gg 1,
\end{equation}
as is given in \cite{RGarrappa16, KGorska21c,RHilfer02a}.

\subsection{Evolution equation for  $n_{H\!N}(\alpha, \beta; t)$}\label{5.5}

According to the general formalism developed in Subsect. \ref{sub3.1}, the relaxation function $n_{H\!N}(\alpha, \beta; t)$ satisfies the integro-differential equations \eqref{14/10-4} and \eqref{13/09-1a}. They contain the memory functions $M_{H\!N}(\alpha, \beta; t)$ and $k_{H\!N}(\alpha, \beta; t)$ connected by the Sonine equation \eqref{14/10-5}. The memory functions $M_{H\!N}$ and $k_{H\!N}$ in the Laplace domain can be determined algebraically using Eqs. \eqref{21/09-1a} and \eqref{28/09-1}: 
\begin{multline}\label{22/06-1}
\widehat{M}_{H\!N}(\alpha, \beta; z) = {\cal B}^{-1}\{[1 + (z\tau)^{\alpha}]^{\beta} - 1\}^{-1} \quad \text{and} \\ \widehat{k}_{H\!N}(\alpha, \beta; z) = \frac{{\cal B}}{z} \{[1 + (z\tau)^{\alpha}]^{\beta} - 1\}.
\end{multline}
The Laplace inversion with respect to time $t$ domain gives
\begin{equation}\nonumber
M_{H\!N}(\alpha, \beta; t) = ({\cal B}t)^{-1} \sum_{r\geq 0} \big(\ulamek{t}{\tau}\big)^{\alpha\beta(1+r)} E_{\alpha, \alpha\beta(1+r)}^{\beta(1+r)}\big[\!-\big(\ulamek{t}{\tau}\big)^{\alpha}\big]
\end{equation}
and
\begin{equation}\nonumber
k_{H\!N}(\alpha, \beta; t) = {\cal B}\big(\ulamek{\tau}{t}\big)^{\alpha\beta}E_{\alpha, 1-\alpha\beta}^{-\beta}\big[\!-\big(\ulamek{t}{\tau}\big)^{\alpha}\big] - {\cal B},
\end{equation}
and for details see Refs. \cite{KGorska21CNSNCa, KGorska21a, AAKhamzin14}. As is shown in Subsect. \ref{sub3.1} the solutions of Eqs. \eqref{14/10-4} and/or \eqref{13/09-1a} are equivalent. Thus, we can choose the equation more convenient for us. Equation \eqref{13/09-1a} contains less complicated memory kernel $k_{H\!N}(\alpha, \beta; t)$ and immediately leads to the evolution equation proposed either in \cite[Eq. (3.40)]{RGarrappa16} or in \cite[Eq. (33)]{KGorska21CNSNCa}:
\begin{equation}\label{22/06-4}
{^{C\!}(D^{\alpha} + \tau^{-\alpha})^{\beta}} n_{H\!N}(\alpha, \beta; t) = - \tau^{-\alpha\beta}.
\end{equation}
The pseudo-operator on the LHS above belongs to the class of Prabhakar-like integral operators, and for the case under consideration is defined as \cite[Eq. (B.23)]{RGarrappa16}
\begin{multline}\label{22/06-5}
{^{C\!}(D^{\alpha} + \tau^{-\alpha})^{\beta}} n_{H\!N}(\alpha, \beta; t) \\ = \int_{0}^{t} (t - \xi)^{-\alpha\beta} E_{\alpha, 1-\alpha\beta}^{-\beta}\big[\!-\!\big(\ulamek{t-\xi}{\tau}\big)^{\alpha}\big] \dot{n}_{H\!N}(\alpha, \beta; \xi) \D\xi,
\end{multline}
where we keep the notation of \cite[Appendix B]{RGarrappa16} and denote $\dot{n}_{H\!N}(\alpha, \beta; t) = \D {n}_{H\!N}(\alpha, \beta; t)/\!\D t$. It has to be pointed out that the pseudo--operator \eqref{22/06-5} must be distinguished from the pseudo--operator ${(^{\rm c\!}D^{\alpha} + \tau^{-\alpha})^{\beta}}$ considered in \cite{RGarrappa16} as an alternative to the LHS of Eq. \eqref{22/06-5}, with the Caputo derivative ${^{\rm c\!}D}^{\alpha}$ (see \cite{IPodlubny99}) sitting in. The difference is evidently seen if we take $\beta=1$, which simplifies HN relaxation to the CC model. Then, it can be concluded that ${^{C\!}(D^{\alpha} + \tau^{-\alpha})}n_{CC}(\alpha; t) = ({^{\rm c\!}D}^{\alpha} + \tau^{-\alpha})n_{CC}(\alpha; t) - \tau^{-\alpha}$, where $n_{CC}(\alpha; t) = n_{H\!N}(\alpha, 1; t)$.

\subsection{Two variants of the subordination approach to  $n_{H\!N}(\alpha, \beta; t)$}\label{sec5.6}

Let us start with the \underline{first} type of subordination, namely Eq. \eqref{30/07-1}, according to which we have
\begin{equation}\label{27/07-2}
n_{H\!N}(\alpha, \beta; t) = \int_{0}^{\infty} n_{D}(\xi) f_{H\!N}(\alpha, \beta; \xi, t) \D\xi,
\end{equation}
where
\begin{multline}\label{25/06-1}
f_{H\!N}(\alpha, \beta; \xi, t) = {\cal B} \E^{{\cal B}\xi}\mathscr{L}^{-1}\Big\{\Big[\frac{\tau^{\alpha\beta}}{z}(\tau^{-\alpha} + z^{\alpha})^{\beta} - \frac{1}{z}\Big] \\ \times \E^{-\xi {\cal B} \tau^{\alpha\beta}(\tau^{-\alpha} + z^{\alpha})^{\beta}}; t\Big\}.
\end{multline}
represents the probability density distribution of (random) internal time $\xi$ with respect to the laboratory time $t$. It is found by using Eq. \eqref{23/02/23-1}. 
Inserting of Eq. \eqref{25/06-1} into \eqref{27/07-2} results in
\begin{multline}\label{8/04-3S}
n_{H\!N}(\alpha, \beta; t) = {\cal B} \int_{0}^{\infty}\! \mathscr{L}^{-1}\Big\{\Big[\frac{\tau^{\alpha\beta}}{z}(\tau^{-\alpha} + z^{\alpha})^{\beta} - \frac{1}{z}\Big] \\ \times \E^{-\xi {\cal B} \tau^{\alpha\beta}(\tau^{-\alpha} + z^{\alpha})^{\beta}}; t\Big\} \D\xi,
\end{multline}
where the term $\exp({\cal B}\xi)$ in Eq. \eqref{25/06-1} canceled $n_{D}(\xi) = \exp(-{\cal B}\xi)$. 

The inverse Laplace transform $\mathscr{L}^{-1}[-, \cdot]$ in the integrand of \eqref{8/04-3S} can be calculated by employing the Efros theorem give the Eq. \eqref{30/03-4} of it with $\widehat{G}_{2}(z) = z^{-1}$ and $\widehat{q}_{2}(z) = \tau^{-\alpha} + z^{\alpha}$. That gives 
\begin{align}\label{8/04-4S}
\mathscr{L}^{-1}\Big\{\frac{1}{z}\Big[\tau^{\alpha\beta}&(\tau^{-\alpha} + z^{\alpha})^{\beta} - 1\Big] \E^{-\xi \mathcal{B} \tau^{\alpha\beta}(\tau^{-\alpha} + z^{\alpha})^{\beta}}; t\Big\} \D\xi \nonumber \\
& = \int_{0}^{\infty} \mathscr{L}^{-1}[(\tau^{\alpha\beta} z^{\beta} - 1) \E^{-\xi \mathcal{B} \tau^{\alpha\beta} z^{\beta}}; u] \nonumber\\
& \times \mathscr{L}^{-1}[z^{-1} \E^{-u (z^{\alpha} + \tau^{-\alpha})}; t] \D u,
\end{align}
where $z\div t$ and $z\div u$ constitute the Laplace pairs for the first and second inverse Laplace transforms,  respectively. The only term which depends on $\xi$ is equal to $\exp(-\xi \mathcal{B} \tau^{\alpha\beta} z^{\beta})$. Thus, inserting Eq. \eqref{8/04-4S} into Eq. \eqref{8/04-3S} and changing the order of integration we get
\begin{align}\label{8/04-5S}
n_{H\!N}(\alpha, \beta; \tau, t) & = \mathcal{B} \int_{0}^{\infty}\! \mathscr{L}^{-1}\Big[(\tau^{\alpha\beta}z^{\beta} - 1) \int_{0}^{\infty}\! \E^{-\xi \mathcal{B} \tau^{\alpha\beta} z^{\beta}}\D\xi; u\Big] \nonumber \\ &\times \mathscr{L}^{-1}[z^{-1}\E^{-u (z^{\alpha} + \tau^{-\alpha})}; t] \D u \nonumber\\
& = \int_{0}^{\infty}\!  \E^{-u \tau^{-\alpha}} \mathscr{L}^{-1}\Big[1 - \frac{1}{\tau^{\alpha\beta} z^{\beta}}; u\Big] \nonumber \\ & \times \mathscr{L}^{-1}[z^{-1}\E^{-u z^{\alpha}}; t] \D u.
\end{align}
Because $\exp(-u \tau^{-\alpha})$ does not depend on $z$ we can extract this term from the second inverse Laplace transform in Eq. \eqref{8/04-5S} and, next, merge it to the first inverse Laplace transform in Eq. \eqref{8/04-5S}. Using the property $\E^{-a t}\mathscr{L}^{-1}[\widehat{g}(z); t] = \mathscr{L}^{-1}[\widehat{g}(z+a); t]$ we get 
\begin{align*}
\E^{-u \tau^{-\alpha}} \!\mathscr{L}^{-1}\Big[1 - \frac{1}{\tau^{\alpha\beta} z^{\beta}}; u\Big] = \mathscr{L}^{-1}\Big[1 - \frac{1}{(1 + \tau^{\alpha}z)^{\beta}}; u\Big] \\ = \delta(u) - \phi_{CD}(\beta; u). 
\end{align*}
The real function $\phi_{C\!D}(\beta; u) \equiv \phi_{H\!N}(1, \beta; u)$ is the response function for the CD model which determines $\phi_{C\!D}(\beta; u) = \mathscr{L}^{-1}[(1 + T s)^{-\beta}; u]$ with $T = \tau^{\alpha}$. Having this in mind we present Eq. \eqref{8/04-5S} as
\begin{align}\label{9/04-1}
n_{H\!N}(&\alpha, \beta; t) = \mathscr{L}^{-1}[z^{-1} \int_{0}^{\infty} \delta(u) \E^{-u z^{\alpha}} \D u; t]  \nonumber \\ & - \int_{0}^{\infty} \phi_{C\!D}(\beta; u) \mathscr{L}^{-1}[z^{-1}\E^{-u z^{\alpha}}; t] \D u \nonumber\\
& = 1 -  \int_{0}^{\infty}\!\! \phi_{C\!D}(\beta; u) \mathscr{L}^{-1}[z^{-1}\E^{-u z^{\alpha}}; t] \D u\nonumber\\
&=1+ \int_{0}^{\infty}\!\! \dot{n}_{C\!D}(\beta; u) \mathscr{L}^{-1}[z^{-1}\E^{-u z^{\alpha}}; t] \D u 
\end{align}
where we used taken with minus sign the standard relation $\phi_{C\!D}(\beta; u) = - \dot{n}_{C\!D}(\beta; u)$ between the response function and the time derivative of relaxation function. Then, we employ the Leibniz formula which allowed us to shift the derivative over $u$ from the relaxation function $n_{C\!D}(\beta; u)$ onto $\mathscr{L}^{-1}[z^{-1}\exp(-u z^{\alpha}); t]$. Thus Eq. \eqref{9/04-1} becomes
\begin{align}\label{10/04-1}
n_{H\!N}(\alpha, \beta; t) & = -\int_{0}^{\infty}\!\! n_{C\!D}(\beta; u) \mathscr{L}^{-1}\Big[z^{-1}\frac{\D}{\D u}\E^{-u z^{\alpha}}; t\Big] \D u \nonumber\\
&= \int_{0}^{\infty}\!\! n_{C\!D}(\beta; u) f(\alpha; u, t) \D u,
\end{align}
where $f(\alpha; u, t)$ can be expressed in terms of the one-sided L\'{e}vy stable distribution $\varPhi_{\alpha}(\sigma)$ with $\alpha\in(0, 1)$ and $\sigma\in\mathbb{R}_{+}$ \cite{KAPenson10}
\begin{align}\label{25/06-1a}
\begin{split}
f(\alpha; u, t) & = \mathscr{L}^{-1}[z^{\alpha - 1} \E^{-u z^{\alpha}}; t] \\ & = \frac{t}{\alpha u^{1 + 1/\alpha}} \varPhi_{\alpha}(t u^{-1/\alpha}).
\end{split}
\end{align}
That means that we arrived at the \underline{second} type of subordination, in which the CD relaxation appears, which plays the role of the parent process and is subordinated by the L\'evy-like process described by the PDF $f(\alpha; u, t)$. This result explicitly confirms our previous claim \cite{KGorska21} that the subordination description of the same process may be realized in different ways which eventually may be understood as an effect of ``nested'' processes.

\section{The Jurlewicz-Weron-Stanislavski model and its physical content}\label{sec6}

The next illustration which shows of applicability of the so far developed universal scheme  presented in Secs. \ref{SEC1}-\ref{SEC3} is the Jurlewicz-Weron-Stanislavsky (JWS) relaxation model. The JWS relaxation model  (for its brief description see \cite{RGarrappa16,AStanislavsky17}) complements and modifies the HN model leaving unchanged its general structure. The model was introduced in \cite{AJurlewicz10,AStanislavsky10,KWeron10} to explain discrepancies emerging in some experiments between the results coming out from the Jonscher URL and those described by the HN pattern, see \cite[Fig. 1]{AStanislavsky17}. (According to statistical enumeration measuring applicability of different models, the JWS pattern fits and reproduces the data for approximately 20$\%$ of relaxation experiments \cite{Jonscher83, KStanislavski16}.) Therefore it is legitimate to analyse this model applying theoretical framework and methods developed in the current work.
 
\subsection{The spectral function $\widehat{\phi}_{J\!W\!S}(\alpha, \beta; \I\!\omega)$}\label{sec6.1}

We shall begin by recalling Eq. \eqref{23/05_2}:
\begin{align}\label{25/06-10}
\begin{split}
\widehat{\phi}_{J\!W\!S}(\alpha, \beta; \I\!\omega) &= \left[\frac{\widehat{\varepsilon^{\star}}(\omega) - \varepsilon_{\infty}}{\varepsilon - \varepsilon_{\infty}}\right]_{J\!W\!S} \\ & = 1 - (\I\!\omega\tau)^{\alpha\beta} \widehat{\phi}_{H\!N}(\alpha, \beta; \I\!\omega) \\
& = 1 - \widehat{\phi}_{H\!N}(-\alpha, \beta; \I\!\omega)
\end{split}
\end{align}
for $\alpha, \beta \in (0, 1]$. For $\alpha = \beta = 1$ it approaches the Debye model whereas for $\alpha\in(0, 1)$ and $\beta = 1$  tends to the CC pattern. The relaxation model with $\alpha = 1$ and $\beta\in(0, 1)$ is called the mirror CD relaxation (MCD) \cite[Fig. 4]{AStanislavsky17}. Note that $\widehat{\phi}_{J\!W\!S}(\alpha, \beta; \I\!\omega)$ tends to $1$ for $\omega\to 0$ and it vanishes in the limit $\omega\rightarrow\infty$. The real and imaginary parts of $\widehat{\varepsilon^{\star}}_{J\!W\!S}(\omega)$ read
\begin{equation}\nonumber
\widehat{\varepsilon'}_{J\!W\!S}(\alpha, \beta; \omega) = \varepsilon - \frac{(\varepsilon - \varepsilon_{\infty}) (\omega\tau)^{\alpha\beta} \cos\big[\beta\,\theta_{\alpha/2}\big((\omega t)^{-2}\big)\big]}{\big[1 + 2(\omega\tau)^{\alpha}\cos(\pi\alpha/2) + (\omega\tau)^{2\alpha}\big]^{\beta/2}}
\end{equation}
and
\begin{equation}\nonumber
\widehat{\varepsilon''}_{J\!W\!S}(\omega) = \frac{(\varepsilon - \varepsilon_{\infty})(\omega\tau)^{\alpha\beta} \sin\big[\beta\,\theta_{\alpha/2}\big((\omega t)^{-2}\big)\big]}{\big[1 + 2(\omega\tau)^{\alpha}\cos(\pi\alpha/2) + (\omega\tau)^{2\alpha}\big]^{\beta/2}},
\end{equation}
respectively, with $\theta_{\alpha}(y)$ introduced in Eq. \eqref{3/06-4}.

\subsection{The response function $\phi_{J\!W\!S}(\alpha, \beta; t)$}\label{sec6.2}

The response function $\phi_{J\!W\!S}(\alpha, \beta; t)$ is the inverse Laplace transform of $\widehat{\phi}_{J\!W\!S}(\alpha, \beta; z)$,
\begin{align}\label{26/06-2}
\phi_{J\!W\!S}(\alpha, \beta; t) & = \mathscr{L}^{-1}[1; t] - \mathscr{L}^{-1}\Big[\frac{z^{\alpha\beta}}{(\tau^{-\alpha} + z^{\alpha})^{\beta}}; t\Big] \nonumber\\
& 
= \delta(t) - \tau^{-1} \big(\ulamek{t}{\tau}\big)^{-1} E_{\alpha, 0}^{\beta}\big(-\big(\ulamek{t}{\tau}\big)^{\alpha}\big),
\end{align}
where we employ the Laplace transform of the Mittag-Leffler function given by Eq. \eqref{15/06-4}. Notice that the $\delta$-Dirac distribution corresponds to the integrable singularity at $t = 0$. We remark that Eq. \eqref{26/06-2} is also presented in \cite[Eq. (3.43)]{RGarrappa16}. For $\alpha\in(0, 1)$ and $\beta = 1$ we obtain the CC model, see Remark \ref{r3-20/02-23}, while the MCD pattern appears for $\alpha = 1$ and $\beta\in(0, 1)$. From Lemma \ref{lem2} we can formally rewrite the JWS response function in terms of the HN response function, namely
\begin{equation}\nonumber
\phi_{J\!W\!S}(\alpha, \beta; t) = \delta(t) - \phi_{H\!N}(-\alpha, \beta; t).
\end{equation}
To rephrase it in terms of functions implemented in the CAS we use Eq. \eqref{26/06-2} and the Mittag-Leffler representation of the response function $\phi_{J\!W\!S}(\alpha, \beta; t)$. This enables one to express $\phi_{J\!W\!S}(\alpha, \beta; t)$ for rational $\alpha = l/k$ in terms of Meijer G-function and the finite sums of generalized hypergeometric functions.   With the help of Eqs. \eqref{17/06-3} and \eqref{16/06-21} we have 
\begin{align}\label{30/06-1} 
\phi_{J\!W\!S}(&l/k, \beta; t) = \delta(t) - (2\pi)^{\frac{1+l}{2} - k}\, \frac{k^{\beta} l^{1/2}}{\tau \Gamma(\beta)} \left(\frac{t}{\tau}\right)^{\!-1} \nonumber \\ & \times G^{k, k}_{k+l, k}\left(\frac{l^{l} \tau^{l}}{t^{l}}\Big\vert {\Delta(k, 1), \Delta(l, 0) \atop \Delta(k, \beta)}\right) \\ \label{31/02/23-10} 
& = \delta(t) - \frac{1}{\tau} \sum_{j=0}^{k-1} \frac{(-1)^{j}(\beta)_{j}}{j!\, \Gamma(\frac{l}{k}j)} \Big(\frac{t}{\tau}\Big)^{\frac{l}{k}j-1} \\ \nonumber 
& \times{_{1+k}F_{l+k}}\left({1, \Delta(k, \beta + j) \atop \Delta(k, 1+j), \Delta(l, \frac{l}{k}j)}; \frac{(-1)^{k} t^{l}}{\tau^{l} l^{l}}\right)\,.
\end{align}
The asymptotic behavior of the Mittag-Leffler functions presented by Lemma \ref{lem1} result in \cite[Eq. (3.45)]{RGarrappa16}, that is
\begin{multline}\nonumber
\phi_{J\!W\!S}(\alpha, \beta; t) \propto \frac{\beta}{\tau \Gamma(\alpha)} \left(\frac{t}{\tau}\right)^{\alpha-1} \quad \text{for} \quad \frac{t}{\tau} \ll 1 \quad \text{and}\\  \phi_{J\!W\!S}(\alpha, \beta; t) \propto - \frac{1}{\tau \Gamma(-\alpha\beta)} \left(\frac{t}{\tau}\right)^{-\alpha\beta-1} \quad \text{for} \quad \frac{t}{\tau} \gg 1.
\end{multline}

Because $\widehat{\phi}_{J\!W\!S}(\alpha, \beta; s)$, $s>0$, has the CM character for $\beta \leq 1/\alpha$, then thanks to the Bernstein theorem (Sect. \ref{subsec3.3}) we know that $\phi_{J\!W\!S}(\alpha, \beta; t)$ is a non-negative function in the same range of parameters $\alpha$ and $\beta$. That is confirmed in Fig. \ref{rys3}.
\begin{figure}[!h]
\begin{center}
\includegraphics[scale=0.42]{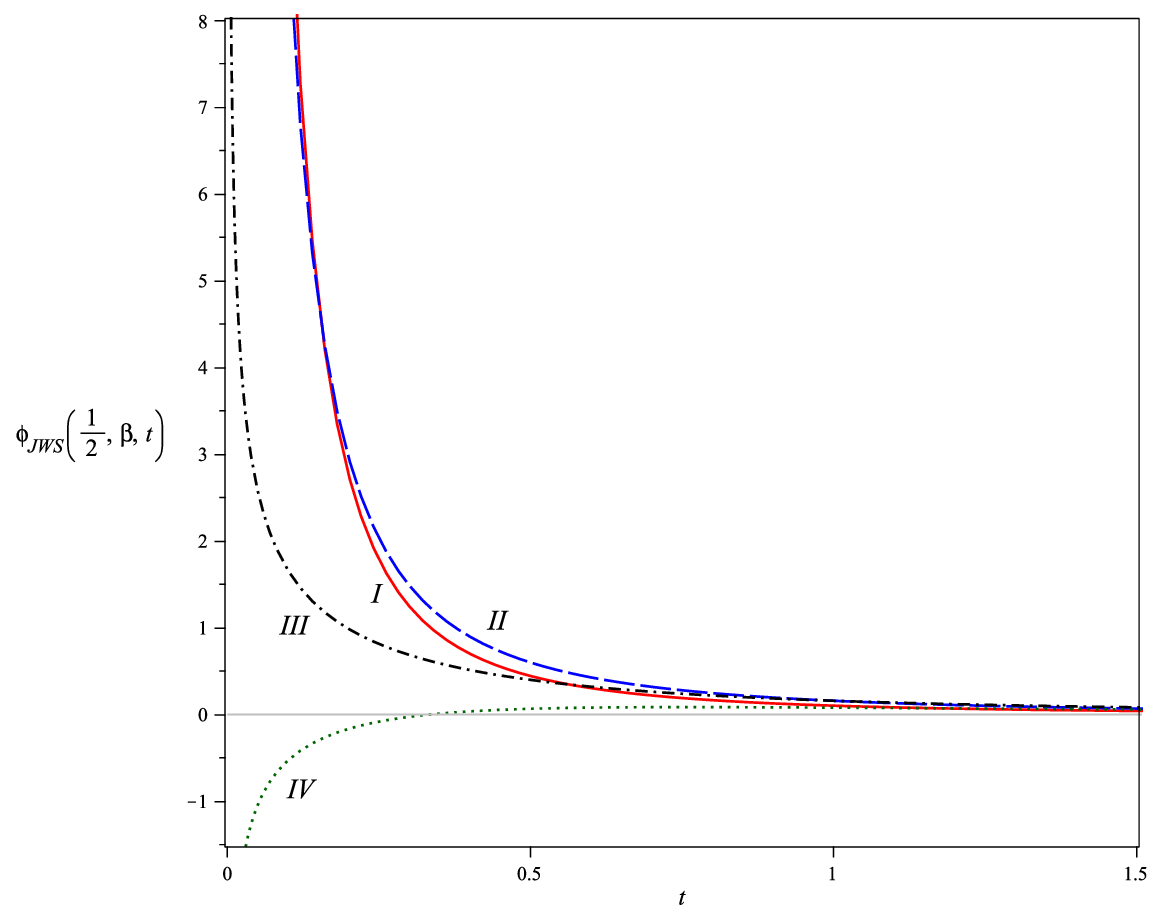}
\caption{\label{rys3}(Color online) Plot of $\phi_{J\!W\!S}(1/2, \beta; t)$ as a function of $t$ given by Eq. \eqref{30/06-1} for $\tau = 1$ and $\beta = 1/2$ (the red continuous curve no. I), $\beta = 1$ (the blue dashed curve no. II), $\beta = 2$ (the black dash-dotted curve no. III), and $\beta = 3$ (the green dotted curve no. IV). } 
\end{center}
\end{figure}

\subsection{The probability density $g_{J\!W\!S}(\alpha, \beta; \xi)$}\label{sec6.3}

By analogy with to Subsect. \ref{sec5.3} we can represent the spectral function $\widehat{\phi}_{J\!W\!S}(\alpha, \beta; z)$, $z=\I\!\omega$, as follows
\begin{equation}\label{30/06-5}
\widehat{\phi}_{J\!W\!S}(\alpha, \beta; z) = \int_{0}^{\infty} \frac{p_{J\!W\!S}(\alpha, \beta; \xi)}{\xi + z} \D\xi,
\end{equation}
which for $z = s \in\mathbb{R}_{+}$  is a S function \cite{KGorska21c}. The Schwinger parametrization enables us to write $\widehat{\phi}_{J\!W\!S}(\alpha, \beta; z) = \mathscr{L}[\phi_{J\!W\!S}(\alpha, \beta; t); z]$, where from Eq. \eqref{18/06-11} it follows that $\phi_{J\!W\!S}(\alpha, \beta; t) = \mathscr{L}[p_{J\!W\!S}(\alpha, \beta; \xi); t]$ with $p_{J\!W\!S}(\alpha, \beta; \xi) = \xi\, g_{J\!W\!S}(\alpha, \beta; \xi)/\tau$. Hence,
\begin{equation}\label{30/06-7}
g_{J\!W\!S}(\alpha, \beta; \xi) = \frac{\tau}{\xi}\, \mathscr{L}^{-1}[\phi_{J\!W\!S}(\alpha, \beta; t); \xi]. 
\end{equation}
Derivation of the exact form of $g_{J\!W\!S}(\alpha, \beta; \xi)$ for $\alpha = l/k$ relies on calculating the inverse Laplace transform sitting in Eq. \eqref{30/06-7}. To do this we employ the Meijer G-representation of the JWS response function. That gives
\begin{multline}\label{2/07-1}
\mathscr{L}^{-1}[\phi_{J\!W\!S}(l/k, \beta; t); \xi] = \mathscr{L}^{-1}[\delta(t); \xi] \\ - (2\pi)^{\frac{1+l}{2} -k} \frac{\sqrt{l} k^{\beta}}{\tau \Gamma(\beta)}\, \mathscr{L}^{-1}\left[\Big(\frac{t}{\tau}\Big)^{-1} \right.\\ \left. \times G^{k, k}_{k+l, k}\left(\frac{l^{l} \tau^{l}}{t^{l}}\Big\vert {\Delta(k, 1), \Delta(l, 0) \atop \Delta(k, \beta)} \right); \xi\right].
\end{multline}
The first inverse Laplace transform in Eq. \eqref{2/07-1},  $\mathscr{L}^{-1}[\delta(t); \xi]$, can be obtained by using the so-called {\em limits representation of $\delta$-distribution}, namely $\delta(t) = \lim_{\epsilon\to +0} \epsilon\, t^{\epsilon - 1}/2$ for $t\in\mathbb{R}_{+}$ \cite{online1}, and next changing the order of the inverse Laplace transform and the limit. That leads to
\begin{align}\nonumber
\begin{split}
\mathscr{L}^{-1}[\delta(t); \xi] & = \mathscr{L}^{-1}\Big[\!\lim_{\epsilon\to +0} \frac{1}{2}\epsilon\, t^{\epsilon - 1}; \xi\Big] \\ &= \frac{1}{2} \lim_{\epsilon\to +0} \epsilon\, \mathscr{L}^{-1}[t^{\epsilon-1}; \xi] \\ & =  \frac{1}{2} \lim_{\epsilon\to +0} \frac{\epsilon\, \xi^{-\epsilon}}{\Gamma(1-\epsilon)} = 0, \qquad \text{for} \quad \xi\in\mathbb{R}_{+}.
\end{split}
\end{align}
To calculate the second inverse Laplace transform we use Eq. \eqref{2/07-11}: 
\begin{align}\label{2/07-3}
g_{J\!W\!S}(&l/k, \beta; \xi) = - (2\pi)^{\frac{1+l}{2} -k} \frac{\sqrt{l} k^{\beta}}{\Gamma(\beta)}\, \frac{\tau}{\xi} \\ & \times \mathscr{L}^{-1}\left[t^{-1}G^{k, k}_{k+l, k}\left(\frac{l^{l} \tau^{l}}{t^{l}}\Big\vert {\Delta(k, 1), \Delta(l, 0) \atop \Delta(k, \beta)} \right); \xi\right] \nonumber \\ \nonumber
& = - (2\pi)^{l-k} \frac{k^{\beta}}{\Gamma(\beta)}\, \frac{1}{\xi}\, G^{k, k}_{k+l, k+l}\left( \xi^{l}\Big\vert {\Delta(k, 1), \Delta(l, 0) \atop \Delta(k, \beta), \Delta(l, 0)} \right).
\end{align}
The last formula becomes more readable  if the Meijer G-function is expressed in terms of the finite sum of the hypergeometric functions ${_{p}F_{q}}$ according to Eq. \eqref{14/07-6}. That implies
\begin{align}\label{4/07-1}
\begin{split}
g_{J\!W\!S}(&l/k, \beta; \xi) = \frac{1}{\pi} \sum_{j=0}^{k-1} \frac{(-1)^{j}}{j!} \frac{(\beta)_{j}}{\xi^{1-\frac{l}{k}n_j}} \sin(\ulamek{l}{k}n_j\pi)\\ & \times  {_{k + 1}F_{k}}\left({1, \Delta(k, n_{j}) \atop \Delta(k, 1+j)}; (-1)^{l-k} \xi^{l}\right),
\end{split}
\end{align}
with $n_j=\beta+j$. This is similar to $g_{H\!N}(l/k, \beta; \xi)$ given by Eq. \eqref{5/06-2} except of terms involving dependence on $\xi$. For the CC model with $\beta = 1$ and $\alpha\in(0, 1)$ we obtain $g_{J\!W\!S}(\alpha, 1; \xi) \equiv g_{CC}(\alpha; \xi)$ given by Eq. \eqref{16/06-1}. For the MCD model, i.e., $\alpha = l/k = 1$, Eq. \eqref{4/07-1} becomes
\begin{equation}\nonumber
g_{J\!W\!S}(1, \beta; \xi) \equiv g_{M\!C\!D}(\beta; \xi) = \frac{\sin(\beta\pi)}{\pi\xi(\xi^{-1} - 1)^{\beta}} \Theta(1-\xi),
\end{equation}
in which the Heaviside step function $\Theta(\cdot)$ ensures that $g_{M\!C\!D}(\beta; \xi)$ is a real function for $0 < \xi < 1$.

Following the same procedure as in the case of the HN model we insert in Eq. \eqref{4/07-1} the series definition of the generalized hypergeometric function ${_{k+1}F_{k}}$, i.e., Eq. \eqref{14/03-A5}. That gives the double sum: one is over $j =0, \ldots, k-1$ and another $r =0, 1, \ldots$. They can be reduced to the single sum due to relation $\sum_{r\geq 0} \sum_{j=0}^{k-1} a_{j + kr} = \sum_{r\geq 0} a_{r}$. Thus, we get an analogue of Eq. \eqref{6/06-4}
\begin{align}\nonumber
g_{J\!W\!S}(&\alpha, \beta; \xi) = \frac{1}{\pi}\, \sum_{r=0}^{\infty} \frac{(-1)^{r}}{r!} \frac{ (\beta)_{r}}{\xi^{1-\alpha(\beta+r)}} \sin[\alpha(\beta+r) \pi] \\ \label{4/07-3}
& = \frac{1}{\pi\xi}\, \IM\left\{\sum_{r=0}^{\infty} \frac{(-1)^{r}}{r!} (\beta)_{r}\, (\xi \E^{\I\!\pi})^{\alpha(\beta + r)}\right\},
\end{align}
where $(\xi \E^{\I\!\pi})^{\alpha(\beta + r)}$ can be defined by the integral representation of $\Gamma$ function, i.e $(\xi \E^{\I\!\pi})^{\lambda} = - \int_{0}^{\infty} \E^{-\xi u} (u \E^{\I\!\pi})^{\lambda - 1} \D u/\Gamma(-\lambda)$. Thereafter, we substitute it into Eq. \eqref{4/07-3} and change the order of sum and integral. Employing the series-form definition of the three-parameter Mittag-Leffler function \eqref{15/06-1a} leads to
\begin{align}\label{4/07-5}
g_{J\!W\!S}(&\alpha, \beta; \xi) = \frac{1}{\pi \xi} \IM\left\{\E^{-\I\alpha\beta\pi} \int_{0}^{\infty} \E^{-\xi u} u^{-\alpha\beta - 1}\right. \nonumber \\ & \left.\times E_{-\alpha, - \alpha\beta}^{\beta}(-u^{-\alpha} \E^{-i\alpha\pi}) \D u\right\} \nonumber \\
& = \frac{1}{\pi \xi} \IM\left\{\int_{0}^{\infty} \E^{-\xi u} u^{-1} E_{\alpha, 0}^{\beta}(-u^{\alpha}\E^{\I\!\alpha\pi}) \D u\right\},
\end{align}
where the relation between two Prabhakar functions comes from the Lemma \ref{lem2}. Inserting the Laplace transform Eq. \eqref{15/06-4} into  Eq. \eqref{4/07-5} and employing de Moivre's formula to calculate the imaginary part  we have 
\begin{equation}\label{5/07-5}
g_{J\!W\!S}(\alpha, \beta; \xi) = \pm\frac{1}{\pi}\, \frac{\xi^{\alpha\beta - 1} \sin[\beta \,\theta_{\alpha}(\xi^{-1})]}{[\xi^{2\alpha} + 2 \xi^{\alpha} \cos(\pi\alpha) + 1]^{\beta/2}},
\end{equation}
where $\theta_{\alpha}(y)$ is defined in Eq. \eqref{3/06-4}. Repeating the considerations quoted below Eq. \eqref{14/06-2} we can deduce that $g_{J\!W\!S}(\alpha, \beta; \xi)$ is non-negative only for $0 < \alpha \leq 1$ and $0 < \beta \leq 1/\alpha$ - that is seen in Figs. \ref{rys4}. In order to avoid dividing the interval $\xi$ into two intervals like in $g_{H\!N}(\alpha, \beta; \xi)$ we illustrate $g_{J\!W\!S}(\alpha, \beta; \xi)$ by employing its Meijer G-representation given by Eq. \eqref{2/07-3}. We have two plots, Fig. \ref{rys4}a and Fig. \ref{rys4}b, where $g_{J\!W\!S}$ are presented as the function of $\xi$ for given $\alpha$ and $\beta$. In Fig. \ref{rys4}a we keep the value of $\alpha = {\rm const}$ and change $\beta$ whereas in Fig. \ref{rys4}b we consider opposite situation.
\begin{figure}[!h]
\includegraphics[scale=0.4]{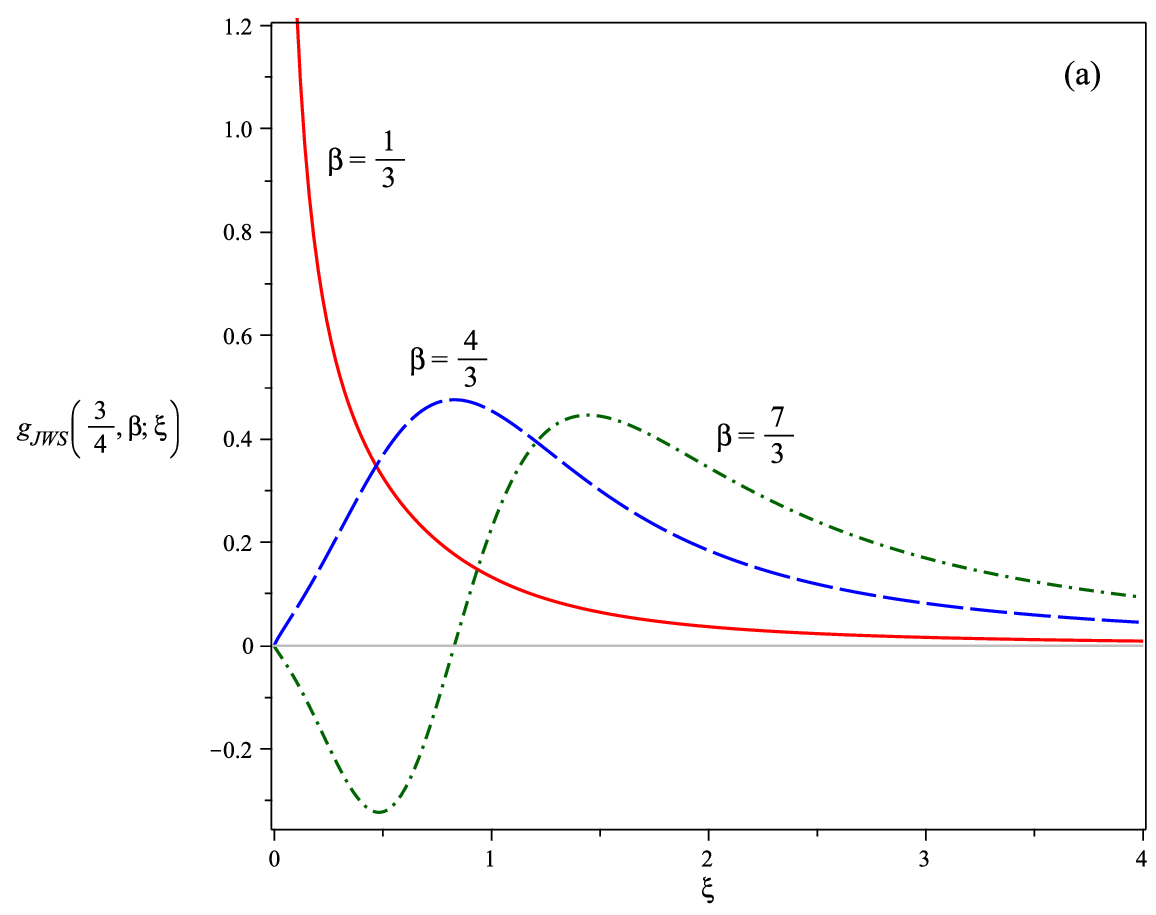}
\includegraphics[scale=0.4]{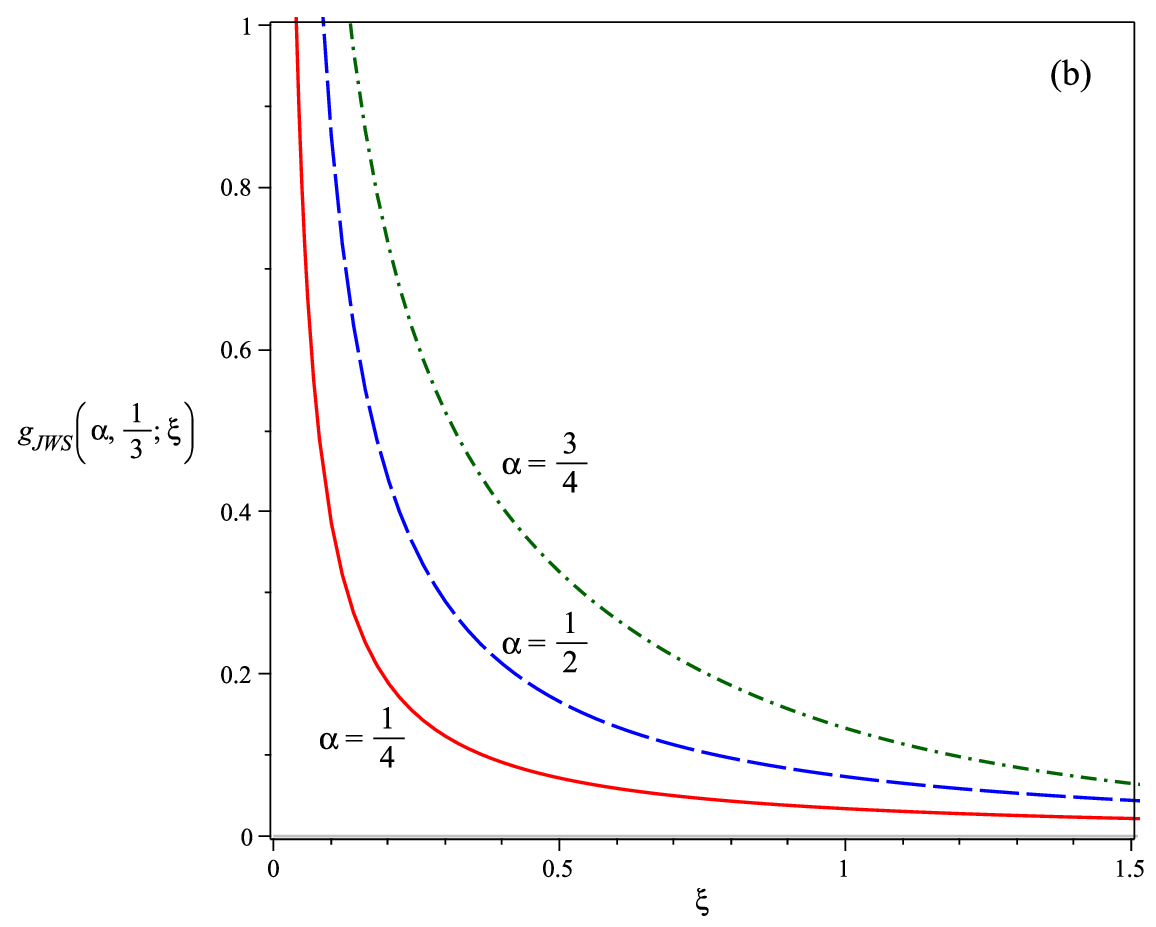}
\caption{\label{rys4}(Color online) Plot of $g_{J\!W\!S}(\alpha, \beta; \xi)$ given by Eq. \eqref{2/07-3} for given parameters $\alpha$ and $\beta$ as the function of $\xi$. Descriptions are the same as in Figs. \ref{rys5}.
} 
\end{figure}

\subsection{The relaxation function $n_{J\!W\!S}(\alpha, \beta; t)$}\label{sec6.4}

Analogically to the HN case the relaxation function for the JWS model can be obtained in least in two ways: either inserting  $g_{J\!W\!S}(\alpha, \beta; \xi)$ found in Subsect. \ref{sec6.3} into Eq. \eqref{18/06-10} or using Eq. \eqref{23/05_1} which connects the Laplace form of the relaxation function $\widehat{n}_{J\!W\!S}(\alpha, \beta; \I\!\omega)$ with the spectral function $\widehat{\phi}_{J\!W\!S}(\alpha, \beta; \I\!\omega)$. Following the second approach we get \cite[Eq. (3.44)]{RGarrappa16}:
\begin{equation}\label{7/07-1}
n_{J\!W\!S}(\alpha, \beta; t) = E_{\alpha, 1}^{\beta}\big(-(t/\tau)^{\alpha}\big),
\end{equation} 
which after using Lemma \ref{lem3-1} and Eq. \eqref{16/06-20} can be rewritten as
\begin{align}\label{1/02/23-1}
\begin{split}
n_{J\!W\!S}(\alpha, \beta; t) & =  \int_{0}^{t} \phi_{H\!N}(-\alpha, \beta; u) \D u \\
& = 1 - n_{H\!N}(-\alpha, \beta; t).
\end{split}
\end{align}
Moreover, with the help of Eqs. \eqref{17/06-3} we can express Eq. \eqref{7/07-1} for $\alpha = l/k$ in the language of the Meijer G-representation:
\begin{align}\label{1/02-2}
n_{J\!W\!S}(&\alpha, \beta; t) = (2\pi)^{\frac{1+l}{2} - k}\, \frac{k^{\beta} l^{-1/2}}{ \Gamma(\beta)} \nonumber \\ & \times G^{k, k}_{k+l, k}\left(\frac{l^{l} \tau^{l}}{t^{l}}\Big\vert {\Delta(k, 1), \Delta(l, 1) \atop \Delta(k, \beta)}\right)
\end{align}
plotted for $\tau=1$, $\alpha = 1/2$ and $\beta = 1/4, 1/2, 3/4$ in Fig. \ref{rys6}.
\begin{figure}[!h]
\begin{center}
\includegraphics[scale=0.42]{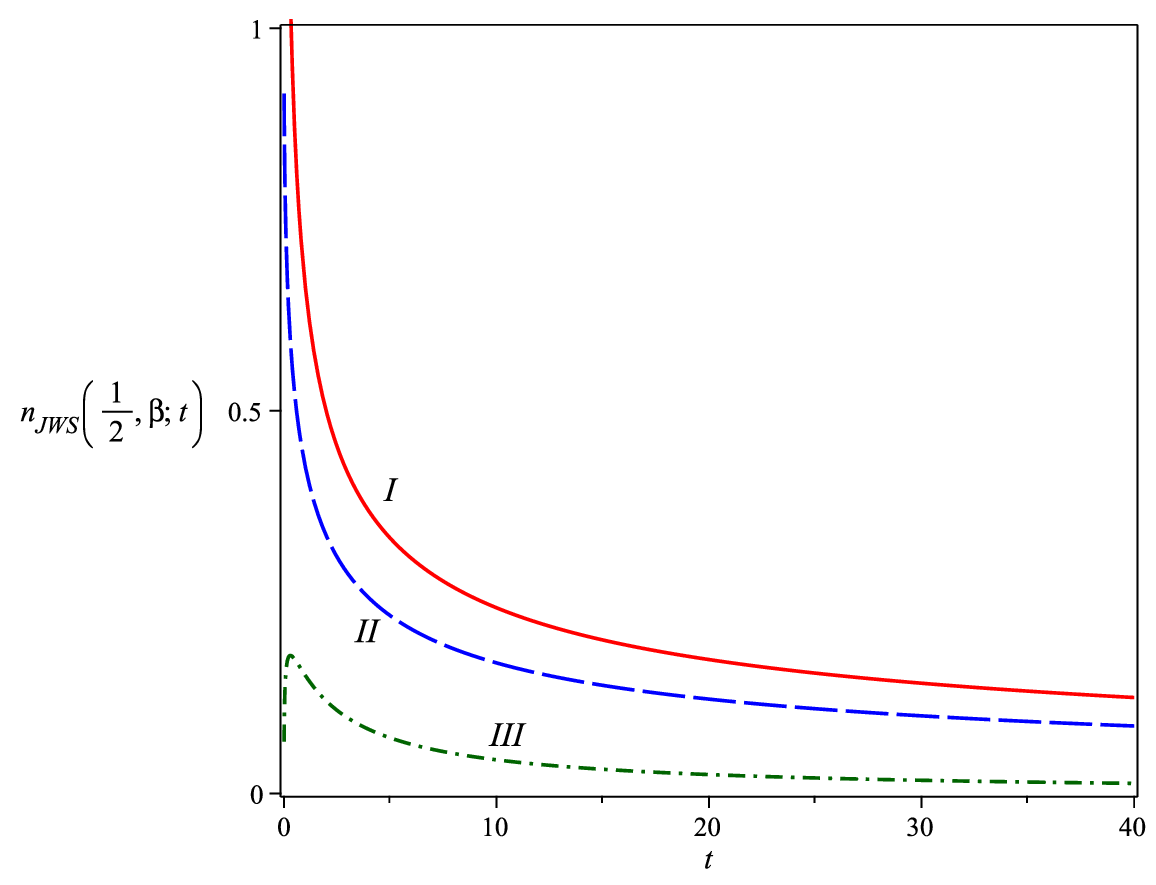}
\caption{\label{rys6}(Color online) Plot of $n_{J\!W\!S}(1/2, \beta; t)$ given by Eq. \eqref{1/02-2}. The curves have analogical characterizations as in Fig. \ref{rys2}.}
\end{center}
\end{figure}{
For $\alpha\in(0, 1)$ and $\beta = 1$ the relaxation function $n_{J\!W\!S}(\alpha, \beta; t)$ goes to the CC relaxation function \eqref{3/08-2} whereas for $\alpha = 1$ and $\beta\in(0, 1)$ it tends to 
\begin{equation}\nonumber
n_{J\!W\!S}(1, \beta; t) \equiv n_{M\!C\!D}(\beta; t) = E_{1, 1}^{\beta}(-t/\tau) = {_{1}F_{1}}\left({\beta \atop 1}; - \frac{t}{\tau}\right).
\end{equation}
The asymptotics of $n_{J\!W\!S}(\alpha, \beta; t)$ obtained from Eqs. \eqref{26/06-10} and \eqref{17/06-25} reads
\begin{multline}\nonumber
n_{J\!W\!S}(\alpha, \beta; t) \propto 1 - \frac{\beta}{\Gamma(1 + \alpha)} \left(\frac{t}{\tau}\right)^{\alpha} \quad \text{for} \quad \frac{t}{\tau} \ll 1 \quad \text{and} \\
n_{J\!W\!S}(\alpha, \beta; t) \propto \frac{1}{\Gamma(1-\alpha\beta)} \left(\frac{t}{\tau}\right)^{-\alpha\beta} \quad \text{for} \quad \frac{t}{\tau} \gg 1.
\end{multline}
like in \cite[Eq. (3.46)]{RGarrappa16}.

\subsection{Evolution equation for  $n_{J\!W\!S}(\alpha, \beta; t)$}\label{sec6.5}

The memory kernels $M_{J\!W\!S}(\alpha, \beta; t)$ and $k_{J\!W\!S}(\alpha, \beta; t)$ obtained with the help of Eqs. \eqref{21/09-1a} and \eqref{28/09-1} furnish
\begin{multline}\nonumber
M_{J\!W\!S}(\alpha, \beta; t) = {\cal B}^{-1} \{t^{-1} E_{\alpha, 0}^{-\beta}[-(t/\tau)^{\alpha}] - \delta(t)\} \quad \text{and} \\ k_{J\!W\!S}(\alpha, \beta; t) = {\cal B} \sum_{r\geq 0}E^{\beta(1+r)}_{\alpha, 1}[-(t/\tau)^{\alpha}],
\end{multline}
see \cite{KGorska21CNSNCa, KGorska21b},  and in the Laplace domain they read
\begin{equation}\nonumber
\widehat{M}_{J\!W\!S}(\alpha, \beta; z) = {\cal B}^{-1}[z^{-\alpha\beta}(\tau^{-\alpha}+z^{\alpha})^{\beta}-1]
\end{equation}
and
 \begin{equation}
 \widehat{k}_{J\!W\!S}(\alpha, \beta; z) = {\cal B}\{z[1+(\tau z)^{-\alpha}]^{\beta} - z\}^{-1}.
\end{equation}
The expression $M_{J\!W\!S}(\alpha, \beta; t)$ is simpler so to get the evolution equation for the JWS model we take Eq. \eqref{14/10-4} instead of Eq. \eqref{13/09-1a}. That gives:
   \begin{equation}\label{24/08-5}
      \int_0^t (t-\xi)^{-1} E_{\alpha, 0}^{-\beta}[-\tau^{-\alpha}(t-\xi)^{\alpha}]\, n_{J\!W\!S}(\alpha, \beta; \xi) \D\xi = 1,
\end{equation}
where we have used the sampling property of the Dirac $\delta$--function, $\int_0^t \delta(t - \xi) n(\xi) \D\xi = n(t)$ for $t\in\mathbb{R}_{+}$. Proceeding as in Ref. \cite{KGorska21CNSNCa} we can state that Eq. \eqref{24/08-5} coincides with the equation \cite[Eq. (3.50)]{RGarrappa16}:  
\begin{equation}\nonumber
(D^{\alpha} + \tau^{-\alpha})^{\beta} n_{J\!W\!S}(\alpha, \beta; t) = \frac{t^{-\alpha\beta}}{\Gamma(1-\alpha\beta)}.
\end{equation}
(completed with a suitable initial condition). The pseudo-operator in the above equation reads
\begin{multline}\label{3/08-6}
(D^{\alpha} + \tau^{-\alpha})^{\beta}n_{J\!W\!S}(\alpha, \beta; t) \\= \frac{\D}{\D t}\int_{0}^{t} (t-\xi)^{-\alpha\beta} E_{\alpha, 1 - \alpha\beta}^{-\beta}\Big[-\frac{(t-\xi)^{\alpha}}{\tau^{\alpha}}\Big]\, 
n_{J\!W\!S}(\alpha, \beta; \xi)\, \D\xi.
\end{multline}
It was introduced in \cite{AStanislavsky15, KStanislavski16} and discussed in Appendix B of \cite{RGarrappa16}. For $\beta = 1$ and $\alpha\in(0, 1)$ it gives the evolution equation written in terms of fractional derivative in the sense of Riemann-Liouville derivative $D^{\alpha}$, $\alpha\in(0, 1)$ \cite{IPodlubny99}, used instead of the fractional derivative in the Caputo sense. Using the link between these derivatives, i.e., $({^{c\!}D^{\alpha}}f)(t) = (D^{\alpha} f)(t) - t^{-\alpha}f(0)/\Gamma(1-\alpha)$, we get the same formula.

\subsection{Two kinds of subordination approaches of  $n_{J\!W\!S}(\alpha, \beta; t)$}\label{sec6.6}

As the \underline{first} type of subordination we take the Debye subordination, i.e., Eq. \eqref{30/07-1} with $f_{\hat{\Psi}}(\xi, t)$ given now by
\begin{multline}\label{11/04-2}
f_{J\!W\!S}(\alpha, \beta; \xi, t) = {\cal B} \E^{{\cal B}\xi} \mathscr{L}^{-1}\Big[\frac{z^{\alpha\beta-1}}{(z^{\alpha} + \tau^{-\alpha})^{\beta} - z^{\alpha\beta}}  \\  \times \E^{-{\cal B}\xi\frac{(z^{\alpha} + \tau^{-\alpha})^{\beta}}{(z^{\alpha} + \tau^{-\alpha})^{\beta} - z^{\alpha\beta}}}; t\Big].
\end{multline}
After the cancelation of the Debye relaxation function by $\exp({\cal B}\xi)$ coming from Eq. \eqref{11/04-2} we get
\begin{multline}\label{11/04-3}
n_{J\!W\!S}(\alpha, \beta; x, t) = {\cal B} \int_{0}^{\infty}\! \mathscr{L}^{-1}\left[\frac{z^{\alpha\beta-1}}{(z^{\alpha} + \tau^{-\alpha})^{\beta} - z^{\alpha\beta}} \right. \\ \left. \times \E^{-\xi\frac{{\cal B}(z^{\alpha} + \tau^{-\alpha})^{\beta}}{(z^{\alpha} + \tau^{-\alpha})^{\beta} - z^{\alpha\beta}}}; t\right] \D\xi.
\end{multline}

To calculate the inverse Laplace transform in the integrand of Eq. \eqref{11/04-3} we apply once more the Efros theorem but this time with $\widehat{G}_{2}(z)$ and $\widehat{q}_{2}(z)$ given below Eq. \eqref{8/04-3S}. Thus, the RHS of Eq. \eqref{11/04-3} becomes
\begin{align}\nonumber
\mathscr{L}^{-1}&\left[\frac{z^{\alpha\beta-1}}{(z^{\alpha} + \tau^{-\alpha})^{\beta} - z^{\alpha\beta}} \E^{-\xi\frac{{\cal B}(z^{\alpha} + \tau^{-\alpha})^{\beta}}{(z^{\alpha} + \tau^{-\alpha})^{\beta} - z^{\alpha\beta}}}; t\right] \nonumber\\
& = \int_{0}^{\infty}\! \mathscr{L}^{-1}\left[\frac{(z - \tau^{-\alpha})^{\beta}}{z^{\beta} - (z - \tau^{-\alpha})^{\beta}} \E^{-\xi \frac{{\cal B} z^{\beta}}{z^{\beta} - (z - \tau^{-\alpha})^{\beta}}}; u\right] \nonumber \\ & \times \mathscr{L}^{-1}[z^{-1} \E^{-u (z^{\alpha} + \tau^{-\alpha})}; t] \D u \nonumber%
\end{align}
Inserting it into Eq. \eqref{11/04-3} and changing the order of integration we can simplify the calculations. It results in 
\begin{align}\label{11/04-5}
n_{J\!W\!S}(&\alpha, \beta; x, t) = {\cal B} \int_{0}^{\infty}\!\E^{-u \tau^{-\alpha}} \mathscr{L}^{-1}\left[\frac{(z-\tau^{-\alpha})^{\beta}}{z^{\beta} - (z-\tau^{-\alpha})^{\beta}} \right.\nonumber \\ & \left.\times \int_{0}^{\infty}\! \E^{-\xi\frac{{\cal B}z^{\beta}}{z^{\beta} - (z - \tau^{-1\alpha})}}\D\xi; u\right]\, \mathscr{L}^{-1}[z^{-1}\E^{-u z^{\alpha}}; t] \D u \nonumber \\
& = \int_{0}^{\infty}\! \mathscr{L}^{-1}\left[\frac{(z - \tau^{-\alpha})^{\beta}}{z^{\beta}}; u\right] \E^{-u\tau^{-\alpha}}  \nonumber \\ &  \times \mathscr{L}^{-1}[z^{-1}\E^{-u z^{\alpha}}; t] \D u.
\end{align}
The next observation is 
\begin{align}\label{11/04-6}
\mathscr{L}^{-1}\Big[&\frac{(z-\tau^{-\alpha})^{\beta}}{z^{\beta}}; u\Big] \E^{-u\tau^{-\alpha}}
= \mathscr{L}^{-1}[(z\tau^{\alpha})^{\beta} \widehat{\phi}_{C\!D}(\beta; z); u]\nonumber\\
&= \mathscr{L}^{-1}[1 - \widehat{\phi}_{M\!C\!D}(\beta; z); u],
\end{align}
where we applied Eq. \eqref{25/06-10} for $\alpha = 1$. Then, it turns out from that $\widehat{\phi}_{M\!C\!D}(\beta; z) = 1 - (z\tau^{\alpha})^{\beta} \widehat{\phi}_{C\!D}(\beta; z)$. Since $\mathscr{L}^{-1}[1; u] = \delta(u)$ and the inverse Laplace transform of the spectral function is equal to the response function, then Eq. \eqref{11/04-6} is expressed as $\delta(u) - \phi_{M\!C\!D}(\beta; u)$. Consequently, $n_{J\!W\!S}$ comes out as
\begin{align}\label{12/04-1}
n_{J\!W\!S}(&\alpha, \beta; x, t) 
= \mathscr{L}^{-1}\left[z^{-1}\int_{0}^{\infty} \! \delta(u) \E^{-u z^{\alpha}} \D u; t\right] \nonumber\\
&
- \int_{0}^{\infty}  \phi_{M\!C\!D}(\beta; u) \mathscr{L}^{-1}\left[z^{-1}\E^{-u z^{\alpha}}; t\right] \D u. 
\end{align}
The last steps to complete the calculation are: the use $\phi_{M\!C\!D}(\beta; u) = - \dot{n}_{M\!C\!D}(\beta; u)$, the Leibniz formula, and, finally, rewrite the RHS of Eq. \eqref{11/04-5} as 
\begin{align}\label{12/04-1}
n_{J\!W\!S}(&\alpha, \beta; x, t) = \int_{0}^{\infty} n_{M\!C\!D}(\beta; u) \mathscr{L}^{-1}\left[z^{-1}\frac{d}{d u}\E^{-u z^{\alpha}}; t\right] \D u \nonumber \\
& = \int_{0}^{\infty} n_{M\!C\!D}(\beta; u) \mathscr{L}^{-1}\left[z^{\alpha-1}\E^{-u z^{\alpha}}; t\right] \D u \nonumber \\ 
& = \int_{0}^{\infty} n_{M\!C\!D}(\beta; u) f(\alpha; u, t) \D u
\end{align}
with $f(\alpha; u, t)$ given by Eq. \eqref{25/06-1a}. Equation \eqref{12/04-1} can be interpreted as the \underline{second} type of subordination in which $n_{M\!C\!D}(\beta; u)$ is subordinated by $f(\alpha; u, t)$.

We see that the HN and the JWS relaxation models lead to at least two types of subordinations: one described by Eqs. \eqref{27/07-2} and \eqref{10/04-1} for the HN model, as well as second describe by Eqs. \eqref{11/04-3} and \eqref{12/04-1} for the JWS model. Looking for physical interpretation of this fact enables us to suspect that with the growing complexity of the relaxing system, the simple partition of the process into two components, namely the parent and leading processes is not enough to reflect and understand all properties of the system, in particular it may be necessary to take into account the fact that the parent and leading processes can have a complex structure on their own.

\section{Summary}

{The broadband dielectric spectroscopy allows us to obtain experimental data which, if fitted to the spectral functions $\widehat{\phi}_{(\cdot)}(\I\!\omega)$, enable us to classify the latter as corresponding to one among of the standard relaxation models. Simpler of them, the CC, CD, and MCD models, emerge as reductions of three parameter ones, namely the HN and JWS patterns. HN and JWS models reduce to the CC relaxation for $\alpha\in(0, 1)$ and $\beta=1$. For  $\alpha = 1$ and $\beta\in(0, 1)$ the HN model goes to CD pattern, whereas the JWS pattern leads to the MCD model. All these spectral functions were tabulated in Tab. \ref{tab-SFs} from which one sees that
\begin{equation}\label{23/02/23-2}
\widehat{\phi}_{J\!W\!S}(\alpha, \beta; \I\!\omega) + \widehat{\phi}_{H\!N}(-\alpha, \beta; \I\!\omega) = 1.
\end{equation}
The knowledge of phenomenologically found spectral functions is crucial for our investigations. From one side, the methods of dielectric relaxation theory and the Laplace transform enable us to recover the response $\phi(t)$ and relaxation $n(t)$ functions, see Eqs. \eqref{23/05_1} and \eqref{16/06-20}, from the knowledge of spectral function. Obviously, the reverse procedure is also justified - we can transform $\phi(t)$ and/or $n(t)$ into $\widehat{\phi}(\I\!\omega)$. For the readers convenience we itemized $\phi(t)$ and $n(t)$ in Tabs. \ref{tab-RFs} and \ref{tab-RelFs} which give
\begin{equation}\label{23/02/23-3}
\phi_{J\!W\!S}(\alpha, \beta; t) + \phi_{H\!N}(-\alpha, \beta; t) = \delta(t)
\end{equation}
and
\begin{equation}\label{23/02/23-3a}
n_{J\!W\!S}(\alpha, \beta; t) + n_{H\!N}(-\alpha, \beta; t) = 1.
\end{equation}
From another side the importance of spectral function $\widehat{\phi}(\I\!\omega)$ comes from the fact that its knowledge enables us to find the memory functions $k(t)$ and $M(t)$. The latter are basic objects used to determine the time evolution of the system under study, see Eq. \eqref{21/09-1a}. To develop this idea we took into account that the memory functions $M(t)$ and $k(t)$ are mutually related by the Sonine equation. This implied that the relevant evolution equations, see Eqs. \eqref{13/09-3}/\eqref{14/10-4} and \eqref{13/09-1a}, led to the same results. Thus we concluded that the search for principles which govern the evolution may be done in two free chosen ways. The simplest is to consider integro-differential equations interpreted as memory dependent. For standard relaxation function these equations are presented in Tab. \ref{tab-EEq}. Here, we would like to pay the readers attention that this approach prefers the evolution equations for the HN and JWS models in the form which involves the pseudo-operators ${^{C\!}(D^{\alpha} + \tau^{-\alpha})^{\beta}}$ and ${(D^{\alpha} + \tau^{-\alpha})^{\beta}}$. These pseudo-operators are defined by Eq. \eqref{22/06-5} and  Eq. \eqref{3/08-6} and differ only by the position of the time derivative, like it is in the fractional derivatives of the Caputo and Riemann-Liouville senses. The relation between them is given though \cite[Eqs. (B.24) or (B.25)]{RGarrappa16}, this is
\begin{multline}\label{22/02/23-1}
{^{C\!}(D^{\alpha} + \tau^{-\alpha})^{\beta}} n(t) \\ = {(D^{\alpha} + \tau^{-\alpha})^{\beta}} n(t) - t^{-\alpha\beta} E_{\alpha, 1-\alpha\beta}^{-\beta}\big(-(\ulamek{t}{\tau})^{\alpha}\big) n(0+).
\end{multline}
The second observation concerning the spectral functions is that they are involved in the definition of characteristic functions $\widehat{\Psi}(z)$, named also the Laplace-L\'{e}vy exponents. This correspondence goes through using the memory function $M(t)$ and appears to be essential if one links the memory dependent evolution schemes  and  subordination approach, see Eq. \eqref{23/02/23-1}. Our approach introduces an essential, previously unnoticed, novelty - in the dielectric relaxation theory we can distinguish, at least two, subordination patterns. Besides of the subordination which comes from the Schwinger parametrization (called by us the basic one), see Eq. \eqref{30/07-1}, we found its alternative coming from the Efros theorem. Both these subordinations are connected to the various choice of ``internal'' timing, see Tab. \ref{tab-sub}. For instant, the HN relaxation function $n_{HN}(t)$ can be build for from the Debye relaxation in which we introduce internal time $\xi$. The relation between the physical time $t$ and internal time $\xi$ is given by the PDF of leading process, here $f_{\widehat{\Psi}(z)}(\xi, t)$. Another possibility to get $n_{HN}(t)$ is from the CD relaxation function in which the timing characterized by the PDF $f(\alpha; \xi, t)$.}
\begin{widetext}
\begin{center}
\begin{table}[!h]
\begin{tabular}{c | c | c | c | c}
$\widehat{\phi}_{HN}$ &  $\widehat{\phi}_{JWS}$ & $\widehat{\phi}_{CC}$ & $\widehat{\phi}_{CD}$ & $\widehat{\phi}_{MCD}$ \\ \hline \hline 
$[1+(\I\!\omega\tau)^{\alpha}]^{-\beta}$ & $1 - [1 + (\I\!\omega\tau)^{-\alpha}]^{-\beta}$ & $[1 + (\I\!\omega\tau)^{\alpha}]^{-1}$ & $(1 + \I\!\omega\tau)^{-\beta}$ & $1- [1 + (\I\!\omega\tau)^{-1}]^{-\beta}$ 
\end{tabular}
\caption{\label{tab-SFs} The spectral functions $\widehat{\phi}(\I\!\omega)$ of the various type of relaxation models, namely the HN, JWS, CC, CD, and postulated MCD models.}
%
\begin{tabular}{c | c | c }
& Mittag-Leffler function & Meijer G-function \\ \hline\hline
$\phi_{H\!N}$ & $\tau^{-1} \big(\frac{t}{\tau}\big)^{\alpha\beta - 1} E^{\beta}_{\alpha, \alpha\beta}\big(\!-\big(\frac{t}{\tau}\big)^{\alpha}\big)$ & $(2\pi)^{\ulamek{1+l}{2} - k} \ulamek{\sqrt{l} k^{\beta}}{t\Gamma(\beta)} G^{k, k}_{l+k, k}\left(\frac{l^{l}\tau^{l}}{t^{l}}\Big\vert{\Delta(k, 1-\beta), \Delta(l, 0) \atop \Delta(k, 0)}\right)$ \\ \hline
$\phi_{J\!W\!S}$ & $\delta(t) - \tau^{-1} \big(\ulamek{t}{\tau}\big)^{-1} E_{\alpha, 0}^{\beta}\big(-\big(\ulamek{t}{\tau}\big)^{\alpha}\big)$ & $\delta(t) - (2\pi)^{\ulamek{1+l}{2} - k} \ulamek{\sqrt{l} k^{\beta}}{t\Gamma(\beta)} G^{k, k}_{l+k, k}\left(\frac{l^{l}\tau^{l}}{t^{l}}\Big\vert{\Delta(k, 1), \Delta(l, 0) \atop \Delta(k, \beta)}\right)
$ \\ \hline
$\phi_{CC}$ & $\tau^{-1}(t/\tau)^{\alpha - 1} E_{\alpha, \alpha}[-(t/\tau)^{\alpha}]$ & $(2\pi)^{\ulamek{1+l}{2} - k}  \frac{\sqrt{l} k}{t} G^{k, k}_{l+k, k}\left(\frac{l^{l}\tau^{l}}{t^{l}}\Big\vert{\Delta(k, 0), \Delta(l, 0) \atop \Delta(k, 0)}\right)$ \\ \hline
$\phi_{C\!D}$ & $\tau^{-1} (t/\tau)^{\beta - 1} E_{1, \beta}^{\beta}[-(t/\tau)]$ & $(t/\tau)^{\beta - 1} \E^{-t/\tau}/ [\tau \Gamma(\beta)]$ \\ \hline
$\phi_{M\!C\!D}$ & $\delta(t) - \tau^{-1} \big(\ulamek{t}{\tau}\big)^{-1} E_{\alpha, 0}\big(-\big(\ulamek{t}{\tau}\big)^{\alpha}\big)$ & $\delta(t) + \beta \tau^{\beta} (\tau + t)^{-1-\beta}$ \\ 
\end{tabular}
\caption{\label{tab-RFs} The response functions $\phi(t)$ presented for the various types of relaxation models, namely HN, JWS, CC, CD, and MCD, in the language of Mittag-Leffler function and Meijer G-function.}

\begin{tabular}{c | c | c }
& Mittag-Leffler function & Meijer G-function \\ \hline\hline
$n_{H\!N}$ & $1 - \big(\ulamek{t}{\tau}\big)^{\alpha\beta} E_{\alpha, 1 + \alpha\beta}^{\beta}\big(\!-(\ulamek{t}{\tau})^{\alpha}\big)$ & $1 - (2\pi)^{\frac{1+l}{2} - k}\,\frac{k^{\beta}l^{-1/2}}{\Gamma(\beta)} G_{l+k, k}^{k, k}\left(\frac{l^{l} \tau^{l}}{t^{l}}\Big\vert {\Delta(k, 1-\beta), \Delta(l, 1) \atop \Delta(k, 0)}\right)$ \\ \hline
$n_{J\!W\!S}$ & $E_{\alpha, 1}^{\beta}\big(-(t/\tau)^{\alpha}\big)$ & $ (2\pi)^{\frac{1+l}{2} - k}\,\frac{k^{\beta}l^{-1/2}}{\Gamma(\beta)} G_{l+k, k}^{k, k}\left(\frac{l^{l} \tau^{l}}{t^{l}}\Big\vert {\Delta(k, 1), \Delta(l, 1) \atop \Delta(k, \beta)}\right)$ \\ \hline
$n_{CC}$ & $E_{\alpha}[-(t/\tau)^{\alpha}]$ & $(2\pi)^{\ulamek{1+l}{2} - k}  \frac{\sqrt{l} k}{t} G^{k, k}_{l+k, k}\left(\frac{l^{l}\tau^{l}}{t^{l}}\Big\vert{\Delta(k, 0), \Delta(l, 0) \atop \Delta(k, 0)}\right)$ \\ \hline
$n_{C\!D}$ & $1 - E^{\beta}_{1, 1 + \beta}(-t/\tau)$ & $\Gamma(\beta, t/\tau)/\Gamma(\beta)$  \\ \hline
$n_{M\!C\!D}$ & $E_{1, 1}^{\beta}(-t/\tau)$ & ${_{1}F_{1}}\left({\beta \atop 1}; - \frac{t}{\tau}\right)$  \end{tabular}
\caption{\label{tab-RelFs} It presents the relaxation functions $n(t)$ for the various types of relaxation models, namely HN, JWS, CC, CD, and MCD, in the language of Mittag-Leffler function and Meijer G-function.}
\end{table}
\end{center}
\end{widetext}

\begin{table}
\begin{center}
\begin{tabular}{c | c }
model & evolution equation \\ \hline\hline
HN & ${^{C\!}(D^{\alpha} + \tau^{-\alpha})^{\beta}} n_{H\!N}(\alpha, \beta; t) = - \tau^{-\alpha\beta}$ \\ \hline
JWS & $(D^{\alpha} + \tau^{-\alpha})^{\beta} n_{J\!W\!S}(\alpha, \beta; t) = t^{-\alpha\beta}/\Gamma(1-\alpha\beta)$ \\ \hline
CC & ${^{c}D^{\alpha}}n_{CC}(\alpha; t) = -\tau^{-\alpha} n_{CC}(\alpha; t)$ \\ \hline
CD & ${^{C\!}\big(\frac{\D}{\D t} + \frac{1}{\tau}\big)^{\!\beta}} n_{C\!D}(\beta; t) = - \tau^{-\beta}$ \\ \hline
MCD & $\big(\frac{\D}{\D t} + \frac{1}{\tau}\big)^{\!\beta} n_{M\!C\!D}(\beta; t) = t^{-\beta}/\Gamma(1-\beta)$
\end{tabular}
\caption{\label{tab-EEq} Table \ref{tab-EEq} presents the evolution equation for the various types of relaxation models, namely HN, JWS, CC, CD, and MCD. The relaxation function $n(t)$ at the initial time $t_{0} = 0$ is equal to 1. The pseudo-operators ${^{C\!}(D^{\alpha} + \tau^{-\alpha})^{\beta}}$ as well as $(D^{\alpha} + \tau^{-\alpha})^{\beta}$ are, respectively, defined in Eqs. \eqref{22/06-5} and \eqref{3/08-6}. \ \\}
\end{center}
\end{table}

\begin{table}[!h]
\begin{center}
\begin{tabular}{c | c  c}
model & $h(\xi)$ & $f(\xi, t)$ \\ \hline\hline
HN &  $n_{D}(\xi)$ & $f_{H\!N}(\alpha, \beta; \xi, t)$ given by Eq. \eqref{25/06-1} \\ 
 & $n_{C\!D}(\beta; \xi)$ & $f(\alpha; \xi, t)$ given by Eq. \eqref{25/06-1a} \\ \hline
JWS & $n_{D}(\xi)$ & $f_{J\!W\!S}(\alpha, \beta; \xi, t)$ given by Eq. \eqref{11/04-3} \\
 & $n_{M\!C\!D}(\beta; \xi)$ & $f(\alpha; \xi, t)$ \\ \hline
 CC & $n_{D}(\xi)$ & $f(\alpha; \xi, t)$ \\ 
  & $n_{CC}(\xi)$ & $\delta(t-\xi)$ \\ \hline
 CD & $n_{D}(\xi)$ & $f_{C\!D}(\beta; \xi, t) = f_{H\!N}(1, \beta; \xi, t)$ \\
  & $n_{C\!D}(\beta; \xi)$ & $\delta(t-\xi)$ \\ \hline
  MDC & $n_{D}(\xi)$ & $f_{M\!C\!D}(\beta; \xi, t) = f_{J\!W\!S}(1, \beta; \xi, t)$ \vspace{0.1cm} \\
  & $n_{M\!C\!D}(\beta; \xi)$ & $\delta(t-\xi)$
\end{tabular}
\caption{\label{tab-sub} The $h(\xi)$ and $f(\xi, t)$ appeared in Eq. \eqref{30/07-1} expressing subordination approach. Observe that for HN and JWS model we present two non-trivial types of subordinations.}
\end{center}
\end{table}

\section{Outlook}

{The standard relaxation models exemplified by the HN and JWS patterns depend on three material dependent parameters $\alpha$, $\beta$, and $\tau$, each time adjusted to the experimental data. These models satisfactorily describe relaxation phenomena characterized by one-peak (unimodal) behavior in the frequency domain. However there exist materials which for frequencies $\gtrsim10^{5}$ Hz exhibit either slower decay or multi-peak behavior of polarizability. This means that so far studied simple relaxation models do not  cover all possibilities thoroughly enough. As prospective challenger models used to describe experimentally more complex phenomena we mention the excess wing (EW) model \cite{KGorska21CNSNCb, RHilfer17, RRNigmatulin16} or models involving sums of standards relaxation patterns taken with different parameters \cite{Liu20}.  We remark that the EW model preserves complete monotonicity in the time domain  while the second approach may lead to the lack of this property and consequently to abandon mathematical methods which result in the complete monotonicity concepts and reflect seemingly imposing, but far from complete, physical interpretation of relaxation in terms of summing up the Debye decays.    

We also emphasize that the two-parameters KWW model mentioned in the Introduction (see Eq. \eqref{23/05_3}) must not be discarded as historical and old-fashioned. If we reduce it to the short-time, i.e., power-law, its asymptotics can be successfully used to describe discharge of atypical capacitors in modeling special electric circuits important for electrochemistry \cite{RTTGettens08, MHeari11} and biochemical processes \cite{EHernandez17, EHernandez20, EHernandez21}. It also serves to be a starting point comparison between exponential- and power-like  time behaviors of relaxation functions \cite{Alvarez91, Alvarez93, KGorska21c}.
}

\section*{Acknowledgments}
KG addresses her special thanks to the LPTMC, Sorbonne Universit\'{e} and personally to its director Prof. B. Delamotte for hospitality and help in arranging her stay in Paris. KG acknowledges the financial support provided to her under the Polish-French Programme  "Long-term research visits in Paris 2022" endowed by PAN (Poland) and CNRS (France). \\
KG and AH research was supported by the Polish National Research Centre (NCN) Research Grant OPUS-12 No. UMO-2016/23/B/ST3/01714; KG acknowledges also financial support provided to her by the NCN Grant Preludium Bis 2 No. UMO-2020/39/O/ST2/01563. \\
{ The authors are very grateful to anonymous referees for careful reading the manuscript, remarks, comments, and suggestions which essentially amended our paper and made it much more readable.}

\appendix
\section{Abbreviations}
{The following abbreviations are used in this manuscript:
\begin{align}\nonumber
\begin{split}
\text{KWW} & \quad \text{Kohlrausch-Williams-Watts } \\[-0.2\baselineskip]
 & \quad \text{(stretched exponential function)} \\[-0.2\baselineskip]
\text{CC} & \quad \text{Cole-Cole} \\[-0.2\baselineskip]
\text{CD} & \quad \text{Cole-Davidson} \\[-0.2\baselineskip]
\text{MCD} & \quad \text{mirror Code-Davidson}\\[-0.2\baselineskip]
\text{HN} & \quad \text{Havrilliak-Negami relaxation} \\[-0.2\baselineskip]
\text{JWS} & \quad \text{Jurlewicz-Weron-Stanislavsky} \\[-0.2\baselineskip]
\text{EW} & \quad \text{excess wings}  \\[0.2\baselineskip]
\text{CM} & \quad \text{completely monotonic} \\[-0.2\baselineskip]
\text{B} & \quad \text{Bernstein} \\[-0.2\baselineskip]
\text{CB} & \quad \text{completely Bernstein} \\[-0.2\baselineskip]
\text{S} & \quad \text{Stieltjes} \\[0.2\baselineskip]
\text{FP} & \quad \text{Fokker-Planck} \\[-0.2\baselineskip]
\text{PDF} & \quad \text{probability density function} \\[-0.2\baselineskip]
\text{CAS} & \quad \text{computer algebra systems} \\[0.2\baselineskip]
\text{LHS} & \quad \text{the left-hand side} \\[-0.2\baselineskip]
\text{RHS} & \quad \text{the right-hand side}
\end{split}
\end{align}
}

\section*{References}

  

\begin{thebibliography}{99}
\bibitem{F2}
Because the time $t$ is defined on the positive semi-axis in physical literature, e.g., \cite{CJFBottcher78}, the Laplace transform $\mathcal{L}[f(t), \I\!\omega]$ is used rather than the Fourier transform defined on the positive semi-axis, see \cite[Subsect. 3.2]{LDebnath19}.
\bibitem{F3}
Physical parameters, as quantities coming from the measurement, must be given by rationales.
\bibitem{F4}
Here, and in what follows, we distinguish between the Laplace integral $\int_{0}^{\infty}\exp{(-xt)}f(t) \D t$, $x\in{\mathbb{R_{+}}}$ and the Laplace transform $\int_{0}^{\infty}\exp{(-zt)}f(t)\D t$, $\RE(z)\in\mathbb{R}_{+}$.
\bibitem{F6} 
In the language of FP equations, this means that the FP operator $\mathcal{L}$ reduces to a constant, in what follows called $\mathcal{B}$. Thus, the considered equation is defined only in time $t$ domain.
\bibitem{F7}
The name ``parent process'' is commonly used in anomalous diffusion phenomena. There it is defined on the space-time domain. Here, we employ this name in the relaxation theory where problems are considered in one dimension only, namely in time.
\end{thebibliography}
\end{document}